\RequirePackage{fix-cm}

\documentclass[%
 reprint,
 amsmath,amssymb,
 aps,
pra,
showkeys,
]{revtex4-2}

\usepackage{graphicx}
\usepackage{dcolumn}
\usepackage{bm}

\usepackage{amsthm}
\usepackage{braket}
\usepackage[T1]{fontenc}
\usepackage[utf8]{inputenc}
\usepackage{lmodern}
\usepackage[dvipsnames]{xcolor}
\usepackage{tikz}
\usetikzlibrary{positioning,calc,arrows.meta,shapes.misc,backgrounds,fit}

\usepackage[ruled,vlined,linesnumbered]{algorithm2e}

\usepackage[caption=false]{subfig}

\usepackage{makecell}
\usepackage{booktabs}

\usepackage{hyperref}
\hypersetup{
  colorlinks   = true, 
  urlcolor     = blue, 
  linkcolor    = blue, 
  citecolor    = blue  
}

\newtheorem{lemma}{Lemma}

\DeclareMathOperator*{\argmax}{arg\,max}

\definecolor{Ink}{HTML}{1F2933}
\definecolor{Muted}{HTML}{66717E}
\definecolor{Panel}{HTML}{F8FAFC}
\definecolor{Gain}{HTML}{176B87}
\definecolor{GainDark}{HTML}{0F526A}
\definecolor{GainLight}{HTML}{EAF4F7}
\definecolor{Fail}{HTML}{B33A3A}
\definecolor{FailLight}{HTML}{FBEDEE}
\definecolor{Good}{HTML}{238B57}
\definecolor{GoodLight}{HTML}{ECF7F1}
\definecolor{Check}{HTML}{E67E22}

\definecolor{Static}{HTML}{B42318}
\definecolor{StaticLight}{HTML}{FDECEC}

\definecolor{Dynamic}{HTML}{167C5A}
\definecolor{DynamicLight}{HTML}{EAF7F1}

\definecolor{CheckFill}{HTML}{F3F4F6}

\tikzset{
  panel/.style={draw=black!18,fill=Panel,rounded corners=2mm,line width=0.55pt},
  head/.style={font=\sffamily\bfseries\small,text=Ink},
  sub/.style={font=\sffamily\bfseries\scriptsize,text=Ink},
  tiny/.style={font=\sffamily\scriptsize,text=Ink},
  micro/.style={font=\sffamily\fontsize{6.6}{7.4}\selectfont,text=Ink},
  note/.style={font=\sffamily\scriptsize,text=Muted},
  dq/.style={rounded rectangle,rounded rectangle arc length=110,draw=GainDark,fill=GainLight,text=GainDark,minimum width=12.5mm,minimum height=5.8mm,inner sep=1pt,font=\sffamily\scriptsize\bfseries},
  dqbad/.style={rounded rectangle,rounded rectangle arc length=110,draw=Fail,fill=FailLight,text=Fail,minimum width=12.5mm,minimum height=5.8mm,inner sep=1pt,font=\sffamily\scriptsize\bfseries},
  con/.style={circle,draw=Check!85!black,fill=Check,text=white,minimum size=5.1mm,inner sep=0pt,font=\sffamily\fontsize{6.2}{6.2}\selectfont\bfseries},
  coff/.style={circle,draw=Check,fill=white,text=Check!80!black,minimum size=5.1mm,inner sep=0pt,font=\sffamily\fontsize{6.2}{6.2}\selectfont\bfseries},
  edge/.style={draw=black!48,line width=0.5pt},
  score/.style={draw=black!17,fill=white,rounded corners=1.2mm,inner xsep=4pt,inner ysep=3pt,align=left,font=\sffamily\fontsize{6.8}{7.8}\selectfont,text=Ink},
  rankgood/.style={draw=Good!60,fill=GoodLight,rounded corners=1mm,inner xsep=3pt,inner ysep=1.6pt,font=\sffamily\fontsize{6.6}{7.2}\selectfont\bfseries,text=Good},
  rankbad/.style={draw=Fail!55,fill=FailLight,rounded corners=1mm,inner xsep=3pt,inner ysep=1.6pt,font=\sffamily\fontsize{6.6}{7.2}\selectfont\bfseries,text=Fail},
  rankmid/.style={draw=black!35,fill=white,rounded corners=1mm,inner xsep=3pt,inner ysep=1.6pt,font=\sffamily\fontsize{6.6}{7.2}\selectfont\bfseries,text=Ink},
  eqbox/.style={draw=black!18,fill=white,rounded corners=1.4mm,inner sep=3pt,font=\sffamily\footnotesize,text=Ink}
}

\begin{document}

\preprint{APS/123-QED}

\title{A Highly Accurate Fast Decoding Framework for QLDPC codes Accelerated by Noise Perturbation and Ensemble Decoding}

\author{Mainak Bhattacharyya}
\email{mainak23@iiserb.ac.in}
\author{Ankur Raina}%
 \email{ankur@iiserb.ac.in}
\affiliation{%
 Department of Electrical Engineering and Computer Science, Indian Institute of Science Education and Research Bhopal, Bhopal 462066, India\\
}%




\date{\today}

\begin{abstract}
    A well-balanced decoder has been central to the development of modern fault-tolerant quantum computing.
    However, the inherent topologies of quantum error correcting codes can limit the performance of many well-studied decoding algorithms.
    In this work, we introduce Noise Assisted Ensemble Decoding (NAED), a highly accurate decoding framework with a significant advantage in real-time speed.
    NAED constructs an ensemble of Tanner forests, obtained as acyclic subgraphs of the original Tanner graph, and performs exact inference on each Tanner forest using a lightweight dynamic programming algorithm.
    The forest construction is guided by synthetic soft information derived jointly from the measured syndrome and channel statistics, with controlled noise perturbations generating diverse yet informative decoding matrix column orderings for the Tanner forest construction across the ensemble.
    Our benchmark results show that the proposed synthetic soft information-driven construction and inference on the Tanner forests can achieve improved or comparable decoding performances to the state-of-the-art decoding solutions, such as BP+OSD$0$, while also providing orders-of-magnitude improvements in per-round decoding speed under circuit-level noise. 
\end{abstract}

\maketitle


\section{\label{sec:introduction}Introduction}
Fault-Tolerant (FT) quantum computing has played a significant role in achieving meaningful performance when deploying quantum protocols on practical hardware.
QLDPC codes have been central to this development, and it is deemed to be the most suitable family of quantum error correcting (QEC) codes for achieving low-overhead fault tolerance \cite{breuckmann2021quantum, vasi2026quantum}.
Many attractive properties of QLDPC codes include a sparse decoding graph, efficient distance scaling \cite{panteleev2021quantum, panteleev2022asymptotically, leverrier2022quantum}, constant overhead \cite{gottesman2013fault}, and numerous hardware-native constructions \cite{bravyi2024high, xu2024constant, strikis2023quantum, berthusen2025toward}.
The first use of sparse graphs in QEC was proposed by MacKay \emph{et al.} \cite{mackay2004sparse}.
However, the construction of asymptotically good QLDPC codes is hard in its core and heavily depends on the algebraic aspects \cite{vasi2026quantum}.
One of the most celebrated sub-class of these QLDPC codes are the topological codes \cite{kitaev2003fault}.
The substantial popularity of these codes is due to the planar surface codes \cite{dennis2002topological, fowler2009high}, which pave the way of planar layout for encoding the qubits in superconducting devices.
These topological codes including the infamous variants of surface codes, color codes; shows poor encoding rate \cite{bravyi2010tradeoffs}.
An ideal quantum code construction should offer a constant encoding rate and a linearly scaling minimum distance.
In this direction, the first construction benchmark was set by the hypergraph-product (HGP) codes proposed by Tillich and Z\'{e}mor \cite{tillich2013quantum}.
HGP codes have a constant rate and quadratic minimum distance scaling, which has been the state-of-the-art QLDPC construction for almost a decade.
The first near-linear minimum distance QLDPC construction was proposed by Panteleev and Kalachev \cite{panteleev2021quantum}. 
Also, most recently, a family of codes with vanishing encoding rate, namely the bivariate bicycle (BB) codes, has emerged as one of the leading QLDPC codes suitable for superconducting qubit architectures \cite{bravyi2024high, mandelbaum2025ibm}.\\\\
This tremendous development of good QLDPC codes comes with a significant challenge of balanced real-time decoding.
Most classical LDPC error correction frameworks use Belief Propagation (BP) as a reliable decoding algorithm \cite{richardson2018design} due to its low latency, high accuracy, and hardware-friendly implementation \cite{gallager1962low, mackay1996near, kschischang2001factor}. 
However, it unavoidably fails for almost all the general-purpose QLDPC codes \cite{roffe2020decoding, panteleev2021degenerate}.
BP decoding algorithm primarily suffers from the abundant presence of degenerate error configurations \cite{poulin2008iterative, fuentes2021degeneracy}.
These are popularly studied as harmful trapping sets, intrinsic to every QLDPC codes \cite{raveendran2021trapping}.
An enormous amount of studies have been done to improve the performance of this infamous BP decoding.
The most explored solution involves various post-processing techniques such as Ordered Statistics Decoding (OSD) \cite{panteleev2021degenerate},  Localized Statistics Decoding (LSD) \cite{hillmann2024localized}, Ordered Tanner Forest (OTF) \cite{demarti2026almost}, Ambiguity Clustering (AC) \cite{wolanski2024ambiguity}, and Stabilizer Inactivation \cite{du2022stabilizer}.
Most of these decoders show significant improvement in accuracy over the base BP decoding at the cost of increased run-time.
Some other notable works focus on improving the message-passing dynamics of BP decoders.
These decoders require an understanding of the message passing decoding bottlenecks for QLDPC codes \cite{bhattacharyya2025decoding, yin2024symbreak, chytas2025enhanced, bhatnagar2026impulse}, generating a plethora of pre-processing, heuristics \cite{chytas2025collective, pradhan2023learning} and neural-network inspired solutions \cite{ninkovic2024decoding, maan2025machine, blue2025machine}.\\\\
In this work, we present a new direction for a low-complexity decoding framework for general-purpose QLDPC codes using dynamic programming and ensemble decoding.
Our main contribution is a novel decoding framework, namely the Noise Assisted Ensemble Decoding (NAED).
NAED enables an ensemble of parallel decoding instances, each performing exact inference on a carefully constructed Tanner forest derived from the underlying Tanner graph of a QLDPC code.
The exact inference algorithm is central to achieving highly accurate solutions with no redundant computing iterations.
This setup uses a single iteration of upward message passing on the Tanner forest, with a low-cost dynamic program to exactly infer the least-cost error assignment for a measured syndrome.
This significantly reduces the need for intensive iterative computations such as the Belief Propagation (BP) algorithm.
The number of message-passing iterations of BP required to achieve decoding success is determined heuristically in most cases.
Unlike BP, this exact inference setup in NAED does not require multiple iterations to estimate a solution to the syndrome.
Further, our work shows that good quality solutions to inference on the Tanner forest can be obtained from noise-perturbed soft weight ordered decoding matrix column processing for the Tanner forest construction.
These synthetic weights are obtained solely from the channel and syndrome information.
Such scoring methods can be used as synthetic soft informations and are sufficiently powerful to either beat or match the performance that we can obtain by any BP-based initialization to the inference on the forest.
We perform circuit-level simulations to test the performance of NAED and observe improved decoding accuracy compared to the OSD post-processor for the surface code.
We also study the performance of NAED on BB codes and propose a residual syndrome-aware Tanner forest construction along with different NAED setups tuned for general QLDPC decoding.
Our simulation results show the tremendous error resolution capabilities of the NAED framework.
Further, we benchmark the decoding times per error correction round and observe orders of decoding time improvement for NAED compared to OSD.
This motivates us to propose NAED as a more suitable candidate to the real-time decoding in achieving a practical fault-tolerant gain.

\section{Preliminaries}
\subsection{Quantum Error Correcting Codes}
Different types of Quantum Error Correcting codes (QECC) can be fitted into the general framework of stabilizer codes \cite{gottesman1997stabilizer}.
QLDPC codes are the sparse versions of these stabilizer codes and further falls under a broader class of QECCs, namely the Calderbank-Shor-Steane (CSS) class \cite{calderbank1996good, steane1996multiple}.
CSS is a class of stabilizer codes with two exclusively different types of stabilizer operators.
These are $X$ and $Z$-type of stabilizers, which consists of either the Pauli-$X$ or, the Pauli-$Z$ operators.
The stabilizer operators defines the encoded space of a QECC and has a matrix representation, where each row of the matrix is the symplectic representation of each of the stabilizer operators.
These matrices are called the parity check matrices (PCM) of the QECC.
CSS codes assume a special PCM structure:
\begin{align}
    \label{eq:pcm-css}
    H = \begin{bmatrix}
        H_X & \mathbf{0}\\
        \mathbf{0} & H_Z
    \end{bmatrix},
\end{align}
The non-trivial construction of these CSS QECCs follows from the fact that all the stabilizer operators must commute with each other, which puts a constraint of $H_XH^{T}_Z = \mathbf{0}$.\\\\
Errors on a set of qubits either commute or anti-commute with the stabilizer operators.
The set of undetectable errors, which are not part of the stabilizer group form a set of logical operators.
These operators are responsible for the logical operations on the encoded logical state of a QECC and also reflect the error correction capabilities of the QECC.
A QECC in general is represented using the notation $[[n,k,d]]$, where $n$ is the total number of physical qubits used in the encoding of $k$ logical qubits worth of information and $d$ is the minimum distance of the code.

\subsection{The Decoding Problem in QEC}
The stabilizer operators prepare an encoded logical state, which gets affected by the noise.
A decoder tasks itself to recover the encoded state, which requires determining the incurred error on the encoded state.
A simple decoding problem can be formulated analogous to the classical case by determining the maximal probable error satisfying the syndrome obtained from the measurement of the stabilizers, i.e.
\begin{align}
    \hat{E} = \argmax_{E \in \mathcal{G}_n} P(E | s),
    \label{eq:qmld}
\end{align}
where, $\mathcal{G}_n$ is the set of all possible $n$-qubit Pauli operators.
A key understanding of the decoding problem formulation is the fact that any Pauli error can be decomposed as follows \cite{poulin2006optimal, iyer2015hardness}:
\begin{align}
    E = T \cdot L \cdot S, \,\, : T \in \mathcal{T}, L \in \mathcal{L}, S \in \mathcal{S} \,\,,
\end{align}
where, $\mathcal{S}$, $\mathcal{L}$ and $\mathcal{T}$ are the set of stabilizer generators, the group of logical operators and the group of pure errors, respectively.
The group of total errors $\mathcal{T} := \{T\}$, consists of operators that commute with all the logical operators, all the stabilizer generators except one stabilizer generator, with which it anti-commutes \cite{gottesman1997stabilizer}.
This decomposition particularly reduces the decoding problem of Eq. \eqref{eq:qmld} into
\begin{align}
    \hat{E} = T_s \argmax_{L \in \mathcal{L}, S \in \mathcal{S}} P(L, S | T_s),
    \label{eq:qmld-decomposed}
\end{align}
where we use the fact that the knowledge of the syndrome $s$ is equivalent to the knowledge of $T_s$, because $T_s \in \mathcal{T}$ is the only candidate that anti-commutes with a unique set of stabilizer generators that generate the syndrome $s$.\\\\
This decoding problem, popularly known as the Quantum Maximum Likelihood Decoding (QMLD), is the classical extension of the decoding formulation for the stabilizer codes.
However, the degeneracy unique to the quantum codes poses a unique feature to error recovery.
Errors can be classified by $L, T \in \mathcal{L}, \mathcal{T}$ \cite{iyer2015hardness}.
Any error with the same label $L,T$ has the same effect on the encoded state and therefore is a suitable candidate for the error recovery.
This deems the QMLD suboptimal for the decoding of the stabilizer codes.
A more appropriate decoding problem for QECCs is the Degenerate Quantum Maximum Likelihood Decoding (DQMLD) \cite{iyer2015hardness, demarti2024decoding}, defined as:
\begin{align}
    \hat{E} = E \in \left\{T_s\hat{L}_sS, \,\,\forall S \in \mathcal{S}\right\},
    \label{eq:dqmld}
\end{align}
where the DQMLD is concerned with determining the optimal logical equivalence class of errors, i.e.,
\begin{align}
    \hat{L}_s = \argmax_{L \in \mathcal{L}}\sum_{S \in \mathcal{S}}P(L,S | T_s),
\end{align}
and any member from the set in Eq. \eqref{eq:dqmld} is an optimum solution to the DQMLD problem.
A good decoder for quantum codes must address the DQMLD, which the existing classical decoders are unable to target due to the abundant degeneracies of the quantum codes.
In this work, we propose a decoder that approximately addresses the DQMLD and shows improved performance compared to the decoders approximating the QMLD.
However, before presenting our proposal, we first describe prior work on accurate and efficient decoding of quantum codes, along with related concepts we use in our work.

\subsection{Related Works and Motivation}
Our work is primarily based on the fact that inference on a syndrome spanning tree or a forest is exact.
One of the notable uses of this phenomenon in QLDPC decoding is exploited by the Ordered Tanner Forest (OTF) post-processor \cite{demarti2026almost}.
The authors use a modified version of Kruskal's algorithm to eliminate columns of a decoding matrix, yielding the construction of cycle-free ordered Tanner forests. 
Inference on Tanner forests is performed via iterative message-passing algorithms such as BP.
However, the authors mention that the oscillations in the soft output of BP \cite{gong2024toward} may prevent the OTF post-processor from predicting an error that resolves the syndrome.
Further, the number of BP iterations required to achieve the inference outcome remains heuristic and may introduce unnecessary latency overhead.
This heuristic inference nature of BP is due to the highly degenerate sub-structures present in the QLDPC codes, which portray BP soft information or, the soft log-likelihood ratio (LLR) as redundant.
Despite their capabilities of capturing more information on circuit-level error propagation, the BP soft LLRs might not be able to carry any critical information to assist in selecting good decoding matrix columns for the OTF creation, when the errors are supported over degenerate sub-structures in the Tanner graph.
Therefore, if we can capture this imperfect soft information by some other means and perform the inference in a low-latency time, it will be much more optimal.\\\\
Maximum a posteriori (MAP) inference has other algorithmic incarnations.
Viterbi decoding is one such alternative, which had an overwhelming influence on classical communication. 
Ollivier and Tillich showed the first trellis constructions for the stabilizer codes and employed the min-sum Viterbi algorithm \cite{forney2005viterbi} as one of the solutions to the quantum decoding \cite{ollivier2006trellises}.
This showcases an early use of dynamic programming in quantum decoding.
The min-sum Viterbi algorithm is a dynamic programming shortest-path algorithm on a directed acyclic graph.
It iteratively maintains a state for each vertex in the trellis sections \cite{ollivier2006trellises}, computing only the minimum-cost survivor path and its predecessor, from which the globally optimal path is reconstructed by traceback \cite{lou2002implementing}.
Dynamic programming, therefore, carries an essence for constructing an optimal solution for the global problem from optimal solutions of smaller sub-problems.\\\\
In classical communication, optimal inference on trees covers a wide range of problems, from coding theory \cite{kschischang2002iterative} to artificial intelligence \cite{pearl2014probabilistic}.
These algorithms also pose a recursive message passing structure, and the most efficient version of such implementation forms a two-pass structure \cite{wainwright2003tree}.
This two-pass structure, although does not offer an asymptotic advantage over standard BP message passing, it can be very efficient at avoiding redundant computations in unwanted message-passing iterations.
For instance, when iterative message passing encounters structures such as trapping sets, increasing the number of iterations does not yield any improvements \cite{du2024check, raveendran2021trapping, kumar2012two}.
In those circumstances, an efficient two-pass algorithm significantly reduces decoding latency.
\section{NAED, A Fast Decoding Framework}
We now formally introduce the ensemble dynamic decoding framework for circuit-level decoding of the QLDPC codes.
The dynamic algorithm solves an exact inference on a Tanner forest, which is primarily assisted by noise perturbation.
We call this decoding framework the Noise Assisted Ensemble Decoding (NAED).\\\\
In this work, we are primarily concerned with the circuit-level decoding and, therefore, in the rest of the work, we primarily recover the errors from the detector error model (DEM) as our decoding matrix.
The columns and rows of the DEM decoding matrix are interpreted as the error mechanisms and detectors; instead of the simpler notion of data qubits and stabilizer checks, respectively (see Appendix \ref{ap:circuit-simulation} for more details).
The latter interpretation is related to the case of decoding over the standard QLDPC parity check matrix and throughout the manuscript.
We simultaneously mention this alternate correspondence in brackets, i.e. a Tanner forest node corresponding to the decoding matrix column is expressed as column (data qubit) and the nodes corresponding to the rows of the matrix is expressed as detector (check).\\\\
In the following sections, we describe various techniques and decoding components used in the development of NAED.
\definecolor{NPath}{HTML}{0F766E}
\definecolor{NPathFill}{HTML}{E8F7F4}
\definecolor{NTarget}{HTML}{7C3AED}
\definecolor{NTargetFill}{HTML}{F1EAFE}
\definecolor{NBranch}{HTML}{AEB8C4}
\definecolor{NBranchFill}{HTML}{F7F8FA}

\tikzset{
 npathpanel/.style={draw=black!18,fill=Panel,rounded corners=2mm,line width=.55pt},
 nptitle/.style={font=\sffamily\selectfont\scriptsize,text=Ink,align=center},
 npnote/.style={font=\sffamily\tiny,text=Muted,align=center},
 nscheck/.style={rectangle,rounded corners=1pt,draw=NPath!78!black,fill=white,
           minimum width=6.4mm,minimum height=4.8mm,inner sep=.7pt,
           font=\sffamily\fontsize{6.2}{6.2}\selectfont,text=Ink},
 nstarget/.style={nscheck,draw=NTarget,fill=NTargetFill,text=NTarget,line width=.9pt,font=\sffamily\fontsize{6.2}{6.2}\selectfont\selectfont},
 ndcol/.style={circle,draw=NPath!85!black,fill=NPathFill,minimum size=5.4mm,
           inner sep=0pt,font=\sffamily\fontsize{5.8}{5.8}\selectfont\selectfont,text=NPath!85!black},
 nsbranch/.style={rectangle,rounded corners=1pt,draw=NBranch,fill=NBranchFill,
            minimum width=5.8mm,minimum height=4.4mm,inner sep=.5pt,
            font=\sffamily\fontsize{5.4}{5.4}\selectfont,text=Muted},
 ndbranch/.style={circle,draw=NBranch,fill=NBranchFill,minimum size=4.8mm,inner sep=0pt,
            font=\sffamily\fontsize{5.1}{5.1}\selectfont,text=Muted},
 npathedge/.style={draw=NPath,line width=1.05pt},
 nbranchedege/.style={draw=NBranch,line width=.55pt},
 nlegend/.style={font=\sffamily\scriptsize,text=Ink}
}

\begin{figure*}[t]
\centering
\resizebox{0.95\textwidth}{!}{%
\begin{tikzpicture}[x=1cm,y=1cm]

\begin{scope}[shift={(0,14.55)}]
\node[npathpanel,minimum width=8.75cm,minimum height=4.65cm,anchor=north west] at (0,0) {};
\node[nptitle] at (4.375,-.32) {(a) Base weight ordering, $\tau^0=0$};
\node[npnote] at (4.375,-.7) {$7$ selected columns on the target-pair path};
\begin{scope}[yshift=-0.8cm]
\node[nstarget] (a0) at (.80,-1.45) {$s_{139}$};
\node[ndcol]       (a1) at (1.55,-1.45) {$d_{523}$};
\node[nscheck]       (a2) at (2.30,-1.45) {$s_{91}$};
\node[ndcol]       (a3) at (3.05,-1.45) {$d_{328}$};
\node[nscheck]       (a4) at (3.80,-1.45) {$s_{86}$};
\node[ndcol]       (a5) at (4.55,-1.45) {$d_{439}$};
\node[nscheck]       (a6) at (5.30,-1.45) {$s_{77}$};
\node[ndcol]       (a7) at (6.05,-1.45) {$d_{298}$};
\node[nscheck]       (a8)  at (6.05,-2.55) {$s_{84}$};
\node[ndcol]       (a9)  at (5.30,-2.55) {$d_{319}$};
\node[nscheck]       (a10) at (4.55,-2.55) {$s_{89}$};
\node[ndcol]       (a11) at (3.80,-2.55) {$d_{511}$};
\node[nscheck]       (a12) at (3.05,-2.55) {$s_{137}$};
\node[ndcol]       (a13) at (2.30,-2.55) {$d_{694}$};
\node[nstarget] (a14) at (1.55,-2.55) {$s_{185}$};
\draw[npathedge] (a0)--(a1)--(a2)--(a3)--(a4)--(a5)--(a6)--(a7)--(a8)--(a9)--(a10)--(a11)--(a12)--(a13)--(a14);

\node[ndbranch] (ab1) at (.80,-.78) {$d_{526}$};
\node[nsbranch] (ab2) at (1.45,-.48) {$s_{100}$};
\draw[nbranchedege] (a0)--(ab1)--(ab2);
\node[ndbranch] (ab3) at (6.75,-2.55) {$d_{482}$};
\node[nsbranch] (ab4) at (7.45,-2.55) {$s_{132}$};
\draw[nbranchedege] (a8)--(ab3)--(ab4);
\end{scope}
\end{scope}

\begin{scope}[shift={(9.15,14.55)}]
\node[npathpanel,minimum width=8.75cm,minimum height=4.65cm,anchor=north west] at (0,0) {};
\node[nptitle] at (4.375,-.32) {(b) Ensemble member, $\tau^1=0.035$};
\node[npnote] at (4.375,-.7) {$1$ selected columns on the target-pair path};

\begin{scope}[yshift=-0.8cm]
\node[nstarget] (b0) at (2.40,-2.05) {$s_{139}$};
\node[ndcol]       (b1) at (4.375,-2.05) {$d_{808}$};
\node[nstarget] (b2) at (6.35,-2.05) {$s_{185}$};
\draw[npathedge,line width=1.35pt] (b0)--(b1)--(b2);

\node[ndbranch] (bb1) at (2.40,-1.15) {$d_{523}$};
\node[nsbranch] (bb2) at (1.55,-.80) {$s_{91}$};
\draw[nbranchedege] (b0)--(bb1)--(bb2);
\node[ndbranch] (bb3) at (6.35,-2.95) {$d_{694}$};
\node[nsbranch] (bb4) at (7.20,-3.30) {$s_{137}$};
\draw[nbranchedege] (b2)--(bb3)--(bb4);
\end{scope}
\end{scope}

\begin{scope}[shift={(0,9.35)}]
\node[npathpanel,minimum width=8.75cm,minimum height=4.75cm,anchor=north west] at (0,0) {};
\node[nptitle] at (4.375,-.32) {(d) Ensemble member, $\tau^2=0.19$};
\node[npnote] at (4.375,-.7) {$14$ selected columns on the target-pair path};
\begin{scope}[yshift=-0.8cm]
\node[nstarget] (d0) at (.55,-1.18) {$s_{139}$};
\node[ndcol] (d1) at (1.30,-1.18) {$d_{523}$};
\node[nscheck] (d2) at (2.05,-1.18) {$s_{91}$};
\node[ndcol] (d3) at (2.80,-1.18) {$d_{328}$};
\node[nscheck] (d4) at (3.55,-1.18) {$s_{86}$};
\node[ndcol] (d5) at (4.30,-1.18) {$d_{439}$};
\node[nscheck] (d6) at (5.05,-1.18) {$s_{77}$};
\node[ndcol] (d7) at (5.80,-1.18) {$d_{298}$};
\node[nscheck] (d8) at (6.55,-1.18) {$s_{84}$};
\node[ndcol] (d9) at (7.30,-1.18) {$d_{319}$};
\node[nscheck] (d10) at (7.30,-2.18) {$s_{89}$};
\node[ndcol] (d11) at (6.55,-2.18) {$d_{193}$};
\node[nscheck] (d12) at (5.80,-2.18) {$s_{41}$};
\node[ndcol] (d13) at (5.05,-2.18) {$d_{18}$};
\node[nscheck] (d14) at (4.30,-2.18) {$s_{4}$};
\node[ndcol] (d15) at (3.55,-2.18) {$d_{14}$};
\node[nscheck] (d16) at (2.80,-2.18) {$s_{6}$};
\node[ndcol] (d17) at (2.05,-2.18) {$d_{29}$};
\node[nscheck] (d18) at (1.30,-2.18) {$s_{27}$};
\node[ndcol] (d19) at (.55,-2.18) {$d_{164}$};
\node[nscheck] (d20) at (.55,-3.18) {$s_{75}$};
\node[ndcol] (d21) at (1.30,-3.18) {$d_{426}$};
\node[nscheck] (d22) at (2.05,-3.18) {$s_{123}$};
\node[ndcol] (d23) at (2.80,-3.18) {$d_{475}$};
\node[nscheck] (d24) at (3.55,-3.18) {$s_{130}$};
\node[ndcol] (d25) at (4.30,-3.18) {$d_{680}$};
\node[nscheck] (d26) at (5.05,-3.18) {$s_{178}$};
\node[ndcol] (d27) at (5.80,-3.18) {$d_{950}$};
\node[nstarget] (d28) at (6.55,-3.18) {$s_{185}$};
\draw[npathedge] (d0)--(d1)--(d2)--(d3)--(d4)--(d5)--(d6)--(d7)--(d8)--(d9)--(d10)--(d11)--(d12)--(d13)--(d14)--(d15)--(d16)--(d17)--(d18)--(d19)--(d20)--(d21)--(d22)--(d23)--(d24)--(d25)--(d26)--(d27)--(d28);

\node[ndbranch] (db1) at (.55,-.62) {$d_{526}$};
\node[nsbranch] (db2) at (1.15,-.34) {$s_{100}$};
\draw[nbranchedege] (d0)--(db1)--(db2);
\node[ndbranch] (db3) at (5.05,-.62) {$d_{168}$};
\node[nsbranch] (db4) at (5.75,-.34) {$s_{29}$};
\draw[nbranchedege] (d6)--(db3)--(db4);
\end{scope}
\end{scope}

\begin{scope}[shift={(9.15,9.35)}]
\node[npathpanel,minimum width=8.75cm,minimum height=4.65cm,anchor=north west] at (0,0) {};
\node[nptitle] at (4.375,-.32) {(f) Ensemble member, $\tau^3=0.215$};
\node[npnote] at (4.375,-.7) {$6$ selected columns on the target-pair path};

\begin{scope}[yshift=-0.8cm]
\node[nstarget] (f0) at (1.15,-1.45) {$s_{139}$};
\node[ndcol] (f1) at (1.95,-1.45) {$d_{523}$};
\node[nscheck] (f2) at (2.75,-1.45) {$s_{91}$};
\node[ndcol] (f3) at (3.55,-1.45) {$d_{328}$};
\node[nscheck] (f4) at (4.35,-1.45) {$s_{86}$};
\node[ndcol] (f5) at (5.15,-1.45) {$d_{445}$};
\node[nscheck] (f6) at (5.95,-1.45) {$s_{125}$};
\node[ndcol] (f7) at (5.95,-2.55) {$d_{487}$};
\node[nscheck] (f8) at (5.15,-2.55) {$s_{132}$};
\node[ndcol] (f9) at (4.35,-2.55) {$d_{684}$};
\node[nscheck] (f10) at (3.55,-2.55) {$s_{180}$};
\node[ndcol] (f11) at (2.75,-2.55) {$d_{810}$};
\node[nstarget] (f12) at (1.95,-2.55) {$s_{185}$};
\draw[npathedge] (f0)--(f1)--(f2)--(f3)--(f4)--(f5)--(f6)--(f7)--(f8)--(f9)--(f10)--(f11)--(f12);

\node[ndbranch] (fb1) at (1.15,-.78) {$d_{526}$};
\node[nsbranch] (fb2) at (1.75,-.48) {$s_{100}$};
\draw[nbranchedege] (f0)--(fb1)--(fb2);
\node[ndbranch] (fb3) at (6.72,-1.45) {$d_{670}$};
\node[nsbranch] (fb4) at (7.48,-1.45) {$s_{173}$};
\draw[nbranchedege] (f6)--(fb3)--(fb4);

\end{scope}
\end{scope}

\end{tikzpicture}%
}
\caption{
We show how different scaling factors for the applied noise can help an ensemble explore different Tanner paths connecting a given pair of violated detector (stabilizer check) nodes.
We use a sparse detector error model (DEM) matrix of the distance $7$ rotated planar surface code and sample a syndrome at a fixed circuit error rate $p = 0.003$.
An ensemble of Tanner forests having different noise scaling parameters is generated using Eq. \eqref{eq:noisy-var-weight}.
In this figure, we only show the path connecting a target detector node pair $\{s_{139}, s_{185}\}$ from the Tanner forests of some ensemble member instances.
We can see that varying the noise scaling parameter $\tau^b$ across different ensemble instances $b$ produces distinct forests having a unique path connecting the violated detectors $\{s_{139}, s_{185}\}$.
}
\label{fig:naed-surface-noisy-tanner-paths}
\end{figure*}
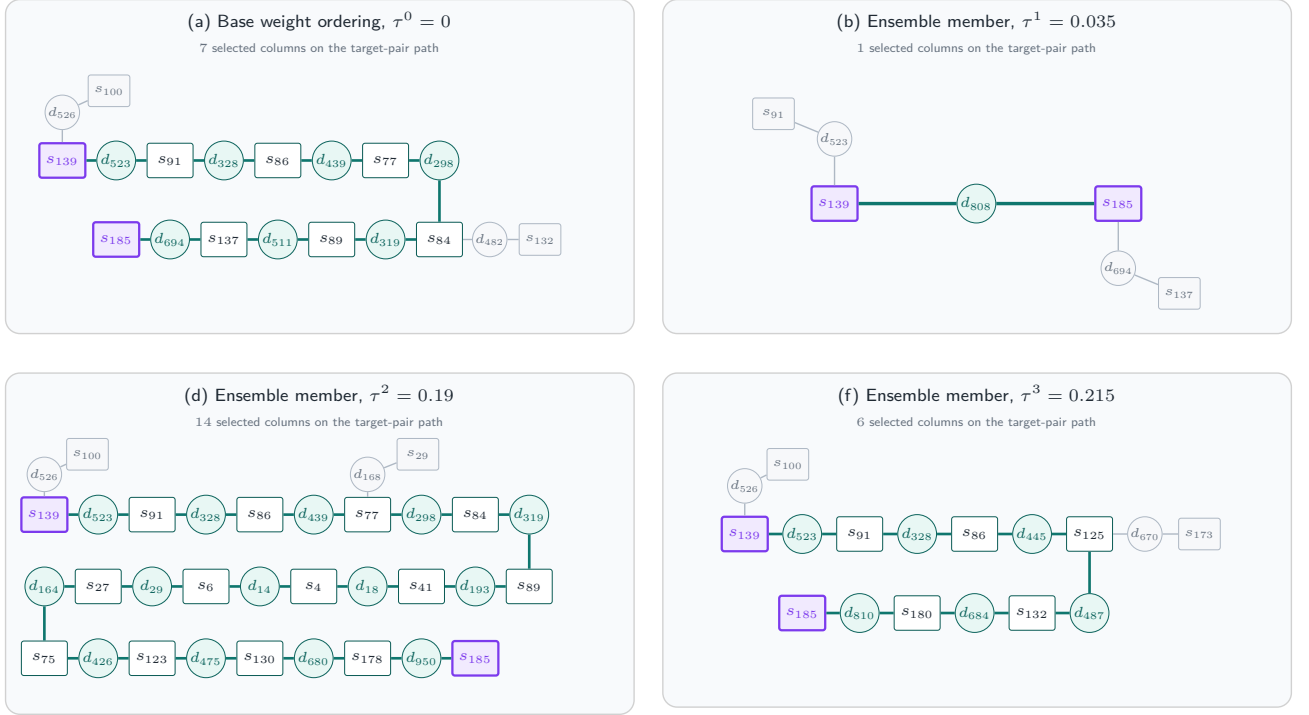
\subsection{Weighted Tanner Forest construction}
The first building block of the decoding pipeline requires construction of the Tanner forests from the degenerate, loopy Tanner graph of the quantum codes.
The Tanner graph can represent both the circuit-level and code-capacity decoding matrix of an underlying QLDPC code.\\
Variants of Kruskal's algorithm have been used in many circumstances to efficiently construct forests from a loopy graph \cite{kruskal1956shortest}.
A notable work includes the ordered Tanner forest construction in the BP+OTF decoder \cite{demarti2026almost}.
The construction of Tanner forest can be achieved in almost linear time using a union-find-type data structure (see Appendix A. of Ref. \cite{demarti2026almost}).
The construction proceeds by an ordered exploration of the columns (data qubits) of the decoding matrix.
Any new column (data qubit) that introduces a loop should be discarded.
After all the troubling columns (data qubits) are discarded, the kept columns (data qubits) form a forest.
The forest might contain several disjoint components.
The overall construction of such Tanner forest is primarily dominated by the sorting process of the columns (data qubit) using a set of weights assigned to them, which is at most $\Theta(n\log(n))$.
We use this method as a primary Tanner forest constructor for the surface codes.
Later, we propose a residual syndrome-aware forest construction, which we apply to the BB codes.
We now discuss the noise perturbed weight assignment of the columns (data qubits), which is exploited in the inference and ordered construction of the Tanner forests.
\subsection{Noisy Weight Assignment of Decoding Matrix Columns}
\label{sec:noisy-dataqubit-weight-assignment}
We assign a base weight to each decoding matrix column (data qubits) $q$, based on their channel LLR values $\mathtt{llr}_q$ and the received syndrome information, which is as follows:
\begin{align}
    w_q = -\mathtt{llr}_q + \alpha \sum_{c \in N(q)} (2s_c - 1),
    \label{eq:base-var-score}
\end{align}
where $N(q)$ is the set of neighbor detector (check) nodes of the column (data qubit) $q$, and $s_c$ is the measured syndrome bit value of detector (stabilizer check) node $c$.
The second term in the above equation is a signed value corresponding to the syndrome of the detector (stabilizer check) $c$.
For instance $2s_c - 1 = +1$, if $s_c = 1$ and it is $-1$, if $s_c = 0$.
The channel LLRs initially prefers the value of columns (data qubits) to be $x_q = 0$, if $\mathtt{llr}_q > 0$. 
This case is valid for all the instances, where the physical error probability of an error mechanism corresponding to the column (data qubit) $q$ adheres to $p_q < \frac{1}{2}$.
The second term serves as a heuristic reliability score based on the syndrome-driven single-column (data qubit) flip gain for column (data qubit) $q$.
We now argue that $w_q$ is a good alternative to the BP-based soft information.\\\\
Consider the decoding problem as an optimization of the cost objective:
\begin{align}
    f(\mathbf{x}) = \alpha\sum_{c=1}^{m}r_c(\mathbf{x}) + \sum_{q=1}^{n}\mathtt{llr}_qx_q,
\end{align}
where $r_c(\mathbf{x})$ is the residual syndrome value of the detector (check) node $c$ and $x_q \in \mathbf{x}$ are the column (data qubit) values.
Now if we flip a column (data qubit) node $q$, then all the syndrome values of the corresponding adjacent detector (check) nodes ($c \in N(q)$) toggle, i.e. $r_c(\mathbf{x}^\prime) = 1 - r_c(\mathbf{x})$, where $\mathbf{x}^\prime$ represents the column (data qubit) values with the $q^{th}$ column (data qubit) value flipped, compared to the initial column (data qubit) values $\mathbf{x}$.
The change in the objective function due to a single flip at column (data qubit) $q$ is
\begin{align*}
     f(\mathbf{x}) - f(\mathbf{x}^\prime) = \alpha \sum_{c \in N(q)}(2r_c - 1) + \mathtt{llr}_q(2x_q - 1).
\end{align*}
Therefore, if we consider the initial condition $\mathbf{x} = \mathbf{0}$ and $r = s$, then the above equation becomes:
\begin{align}
     f(\mathbf{x} = 0) - f(\mathbf{x}^\prime) = \alpha \sum_{c \in N(q)}(2s_c - 1) -\mathtt{llr}_q.
\end{align}
Therefore, we see that the base score $w_q$ has a rigorous interpretation, i.e. $w_q = f(\mathbf{x} = \mathbf{0}) - f(\mathbf{x}^\prime)$.
Now the LLR equivalent of $w_q$ is simply $\lambda_q = -w_q$.
For instance $\lambda_q > 0$ means $f(\mathbf{x}^\prime) > f(\mathbf{x} = \mathbf{0})$; which imply flipping column (data qubit) $q$ increases the objective value.
Therefore, it is preferred to keep the column (data qubit) $q$ in the initial state.
On the other hand, $\lambda_q < 0$ means $f(\mathbf{x}^\prime) < f(\mathbf{x} = \mathbf{0})$, which implies that a flip of column (data qubit) $q$ will further minimize the objective and therefore is a more preferable option.
Therefore, the base weight for each column (data qubit) $w_q$ gives a good estimate of the preferable vulnerable columns (data qubits) from the syndrome and channel information, on a similar note to the purpose that the BP soft LLRs serve.
BP soft LLRs quantify how strongly the decoder must prefer the column (data qubit) $q$ to be flipped.
A positive soft LLR prefers the column (data qubit) to remain in $x_q = 0$, and a negative soft LLR prefers the flip, i.e., $x_q = 1$.
We note that these base scores from Eq. \eqref{eq:base-var-score} do not fully capture the information of BP soft LLRs.
BP soft LLRs can capture much stronger error propagation characteristics when the errors are not supported over any harmful column (data qubit) subsets.
However, when errors are supported on structures like quantum trapping sets, these soft LLRs are highly symmetric and oscillatory.
Under such circumstances, BP soft LLRs do not help infer good solutions for quantum decoding.
In our work, we do not identify BP as a primary source of inference.
Instead, we require a set of soft information to prioritize the columns (data qubits) in the construction of the Tanner forests.
Therefore, the base scores serve as an alternative to the `noisy' information captured by BP under the above-mentioned harmful scenarios.
The main purpose of these base scores $w_q$ from Eq. \eqref{eq:base-var-score} is to identify the critical column (data qubit) nodes for the later stages of inference.
These serve the approximate purpose of prioritizing a set of column (data qubit) nodes that might be involved in the generation of the received syndrome.\\\\
The base score alone is not sufficient to find an optimal column (data qubit) priority set for the construction of a good Tanner forest that either supports the syndrome or a valid logical equivalence class.
The nature of the base score is deterministic.
Consider the case, when a set of columns (data qubits) has been assigned the same base weights, and during the ordered processing of the columns (data qubits) one of the columns (data qubits) could not be prioritized due to the equal weight with other likely columns (data qubits) and a pre-assigned default order for such equal weight column (data qubit) processing.
Under those circumstances, we can never retrieve an essential column (data qubit) into the forest due to deterministic $w_q$.
To overcome such a limitation of the base weight, we need to perturb the base weights for the construction of a diverse profile of forests.
A diverse order can create a diverse forest profile and enable the exact inference procedure to explore different solution choices for the received syndrome, i.e., find a set of column (data qubit) nodes such that $s \in \mathrm{Image}(H_{\mathrm{F}})$.
Once we find a set of columns (data qubits) that satisfy this condition, the inference algorithm will always find the exact solution to the syndrome.
Further, this seeks an implementation of inference on multiple independent forests, which further perfectly fits into the ensemble decoding framework.
On each decoding instance of the ensemble $b$, we perturb the base score for each column (data qubit) by an i.i.d Gaussian noise.
The noise perturbed column (data qubit) scores are as follows: 
\begin{align}
    n^{b}_q = w_q + \tau^{b}\varepsilon_{q},
    \label{eq:noisy-var-weight}
\end{align}
where, $\varepsilon_{q}$ is the i.i.d. sampled Gaussian noise and $\tau^{b}$ is a scaling factor for the applied noise.
The scaling factor is varied across different decoding instances of the ensemble to achieve a diverse profile for the forests, see Fig. \ref{fig:naed-surface-noisy-tanner-paths}. 
For the ordered forest construction, we use $n^{b}_q$ as the sorted priority weights.
On the other hand, $\lambda^{\mathrm{noise}}_q = - n^{b}_q$ is used as the synthetic soft information for the inference stage of the decoding.
Less positive (more negative) $\lambda^{\mathrm{noise}}_q$ encourages prior selection of the corresponding column (data qubit) node $q$ compared to the columns (data qubits) with more positive (less negative) $\lambda^{\mathrm{noise}}_q$.


\begin{figure*}[t]
\centering

\resizebox{0.9\textwidth}{!}{%
\tikzset{
    data/.style={
        circle,
        draw=Ink!75,
        fill=white,
        minimum size=8.4mm,
        inner sep=0pt,
        font=\scriptsize\selectfont,
        text=Ink
    },
    dataStatic/.style={
        data,
        draw=Static,
        fill=StaticLight,
        line width=0.95pt
    },
    dataDynamic/.style={
        data,
        draw=Dynamic,
        fill=DynamicLight,
        line width=0.95pt
    },
    dataMuted/.style={
        data,
        draw=Muted!40,
        fill=white,
        text=Muted!55
    },
    check/.style={
        rectangle,
        rounded corners=1.2pt,
        draw=Ink!70,
        fill=CheckFill,
        minimum width=9.5mm,
        minimum height=6.5mm,
        inner sep=1pt,
        font=\scriptsize,
        text=Ink
    },
    checkHot/.style={
        check,
        draw=Static,
        fill=StaticLight,
        line width=0.8pt
    },
    checkGood/.style={
        check,
        draw=Dynamic,
        fill=DynamicLight,
        line width=0.8pt
    },
    checkMuted/.style={
        check,
        draw=Muted!35,
        fill=white,
        text=Muted!55
    },
    edge/.style={
        draw=Ink!58,
        line width=0.8pt
    },
    cycleedge/.style={
        draw=Static,
        line width=1.6pt
    },
    proposed/.style={
        draw=Static,
        line width=1.0pt,
        dashed
    },
    goodedge/.style={
        draw=Dynamic,
        line width=1.35pt
    },
    mutededge/.style={
        draw=Muted!30,
        line width=0.7pt,
        dashed
    },
    ordernode/.style={
        rounded corners=1.5pt,
        draw=Ink!50,
        fill=white,
        minimum width=13mm,
        minimum height=6.5mm,
        inner sep=1.5pt,
        font=\scriptsize
    },
    orderStatic/.style={
        ordernode,
        draw=Static,
        fill=StaticLight,
        text=Static
    },
    orderDynamic/.style={
        ordernode,
        draw=Dynamic,
        fill=DynamicLight,
        text=Dynamic
    },
    ordarrow/.style={
        ->,
        densely dotted,
        draw=Muted,
        line width=0.95pt
    },
    smallnote/.style={
        font=\scriptsize,
        text=Muted,
        align=center
    },
    paneltitle/.style={
        font=\small\selectfont,
        text=Ink,
        align=center
    },
    callStatic/.style={
        draw=Static,
        fill=StaticLight,
        rounded corners=2pt,
        text=Static,
        font=\scriptsize\selectfont,
        align=center,
        inner sep=3pt
    },
    callDynamic/.style={
        draw=Dynamic,
        fill=DynamicLight,
        rounded corners=2pt,
        text=Dynamic,
        font=\scriptsize\selectfont,
        align=center,
        inner sep=3pt
    }
}
\begin{tikzpicture}[font=\small, >={Latex[length=2.1mm]}]

\begin{scope}[shift={(2.875,6.80)}]

    \node[
        font=\scriptsize\selectfont,
        text=Static,
        anchor=east
    ] at (1.20,-0.10)
    {
        static
    };

    \node[ordernode]
        (o300)
        at (2.05,-0.10)
        {$d_{300}$};

    \node[ordernode]
        (o4961)
        at (3.75,-0.10)
        {$d_{4961}$};

    \node[orderStatic]
        (o214)
        at (5.45,-0.10)
        {$d_{214}$};

    \node[
        font=\large,
        text=Muted
    ] (dotsS) at (7.05,-0.10)
    {
        $\cdots$
    };

    \node[orderStatic]
        (o213s)
        at (8.55,-0.10)
        {$d_{213}$};

    \draw[ordarrow]
        (o300) -- (o4961);

    \draw[ordarrow]
        (o4961) -- (o214);

    \draw[ordarrow]
        (o214) -- (dotsS);

    \draw[ordarrow]
        (dotsS) -- (o213s);

    \node[smallnote, anchor=north]
        at (2.05,-0.46)
        {order $79$};

    \node[smallnote, anchor=north]
        at (3.75,-0.46)
        {order $80$};

    \node[
        smallnote,
        text=Static,
        anchor=north
    ] at (5.45,-0.46)
        {order $81$};

    \node[
        smallnote,
        text=Static,
        anchor=north
    ] at (10,-0.46)
        {order $133$; blocked};

    \node[
        font=\scriptsize\selectfont,
        text=Dynamic,
        anchor=east
    ] at (1.20,-2.40)
    {
        dynamic
    };

    \node[ordernode]
        (o1728)
        at (2.40,-2.4)
        {$d_{1728}$};

    \node[orderDynamic]
        (o213d)
        at (4.30,-2.4)
        {$d_{213}$};

    \node[ordernode]
        (o626)
        at (6.20,-2.4)
        {$d_{626}$};

    \node[ordernode]
        (o2651)
        at (8.10,-2.4)
        {$d_{2651}$};

    \draw[ordarrow]
        (o1728) -- (o213d);

    \draw[ordarrow]
        (o213d) -- (o626);

    \draw[ordarrow]
        (o626) -- (o2651);

    \node[
        smallnote,
        anchor=north
    ] at (2.40,-2.75)
        {step $11$};

    \node[
        smallnote,
        text=Dynamic,
        anchor=north
    ] at (4.30,-2.75)
        {step $12$};

    \node[
        smallnote,
        anchor=north
    ] at (6.20,-2.75)
        {step $13$};

    \node[
        smallnote,
        anchor=north
    ] at (8.10,-2.75)
        {step $14$};

\coordinate (gainmid) at (6.55,-1.50);

\draw[
    ->,
    draw=Dynamic,
    line width=1.15pt
]
    ($(o213s.south)+(0,-0.12)$)
    ..
    controls
        ($(gainmid)+(1.25,0.15)$)
        and
        ($(gainmid)+(-1.25,0.15)$)
    ..
    ($(o213d.north)+(0,0.12)$);

\node[
    callDynamic
] (gainlabel) at (gainmid)
{
    gain promotes $d_{213}$
};


    \node[paneltitle, text width=14cm, anchor=north, align=center]
        at (4.875,-3.62)
    {
        (a) Residual-syndrome-aware gain changes the effective column priority.
    };

\end{scope}


\begin{scope}[shift={(1.475,-2.70)}]

    \node[checkHot] (s25)  at (0.65,2.45) {$s_{25}$};
    \node[checkHot] (s26)  at (4.15,2.45) {$s_{26}$};

    \node[check]    (s114) at (2.40,0.20) {$s_{114}$};
    \node[checkHot] (s42)  at (2.40,4.15) {$s_{42}$};

    %
    \node[data] (s214) at (2.40,1.40) {$d_{214}$};

    \draw[cycleedge]
        (s25) -- (s214) -- (s26);

    \draw[edge]
        (s214) -- (s114);

    \node[
        smallnote,
        anchor=north,
        align=center
    ] at (2.40,-0.1)
    {
        $d_{214}$ is processed and accepted at order $81$
    };

    %
    \node[dataStatic] (s213) at (2.40,3.15) {$d_{213}$};

    \draw[cycleedge]
        (s25) -- (s213) -- (s26);

    \draw[proposed]
        (s213) -- (s42);



    \node[
        font=\scriptsize,
        text=Static,
        anchor=center,
        align=center
    ] at (2.40,4.8)
    {
        $d_{213}$ is processed at order $133$,\\ forms a Tanner cycle
    };


    \node[paneltitle, text width=6.9cm, anchor=north]
        at (2.40,-1.45)
    {
        (b) Static ordering:\\ column $d_{214}$ is accepted first, leaving $d_{213}$ cycle-blocked.
    };

\end{scope}

\draw[
    Muted!25,
    line width=0.75pt
]
    (7.75,1.95)
    --
    (7.75,-4.18);

\begin{scope}[shift={(9.225,-2.70)}]

    \node[checkGood]  (d25)  at (0.65,2.45) {$s_{25}$};
    \node[checkGood]  (d26)  at (4.15,2.45) {$s_{26}$};
    \node[checkGood]  (d42)  at (2.40,4.15) {$s_{42}$};

    \node[checkMuted] (d114) at (2.40,0.20) {$s_{114}$};

    \node[dataDynamic] (d213) at (2.40,3.15) {$d_{213}$};

    \draw[goodedge]
        (d213) -- (d25);

    \draw[goodedge]
        (d213) -- (d26);

    \draw[goodedge]
        (d213) -- (d42);

    \node[
        font=\scriptsize,
        text=Static,
        anchor=center,
        text=Dynamic,
    ] at (2.40,4.7)
    {
        Adding $d_{213} \implies |r|:27\rightarrow24$
    };

    \node[dataMuted] (d214) at (2.40,1.40) {$d_{214}$};

    \draw[mutededge]
        (d214) -- (d25);

    \draw[mutededge]
        (d214) -- (d26);

    \draw[mutededge]
        (d214) -- (d114);



    \node[
        smallnote,
        anchor=north,
        align=center,
        text width=9cm
    ] at (2.40,-0.42)
    {
        $G_{213}=3:
        \rightarrow
        n_{213}
        =
        -1.737 + 2 \times 3
        =
        4.263$;
        \\

        after $d_{213}$ added in F:\\
        $G_{214}=-3:
        \rightarrow
        n_{214}=-1.596-6=-7.596$.
    };

    \node[paneltitle, text width=8.0cm, anchor=north]
        at (2.40,-1.72)
    {
        (c) Dynamic ordering:\\ Column $d_{213}$ jumps forward.
    };

\end{scope}

\end{tikzpicture}%
}

\caption{
In this figure, we show how the dynamic gain term in Eq. \eqref{eq:gain1} helps select good columns of a decoding matrix for general purpose QLDPC codes, where the column weights of the decoding matrix exceeds $2$.
We sample a syndrome from the circuit level noise over the BB code $[[144,12,12]]$ at a fixed circuit noise rate $p = 0.003$ and the decoding happens over the sparse detector error model (DEM) of the BB code.
Under the static case, i.e. $\beta = 0$, the orders of the sorted columns are fixed for the forest construction of each ensemble instance.
In (a) we show the sorted column order of one such instance labeled as `static'.
However, under the dynamic forest construction, we set $\beta = 2$ and the order of the sorted columns changes after each new column inclusion step to the forest $F$.
We show a typical instance of a set of columns that has been added in one of the ensemble instance during this dynamic residual-syndrome aware forest construction in (a) and label it as `dynamic'.
We use $d$ and $s$ to identify the columns and rows of the sparse DEM.
Under the static ordering, column $d_{213}$ is processed at $133^{th}$ order.
This creates a cycle-blocking instance during the forest construction, shown in (b). 
The reason behind this is that another column $d_{214}$ was earlier in the order and therefore got processed and included into the forest earlier during the construction.
In (c) we show that the promotion of the column $d_{213}$ in the processing order helps it being accepted before $d_{214}$ as the priority score of $d_{214}$ is decreased significantly and when it finally gets processed, the cycle blocking applies to $d_{214}$ and gets rejected from the inclusion into forest $F$.
}
\label{fig:static-cycle-dynamic-gain-rescue}

\end{figure*}
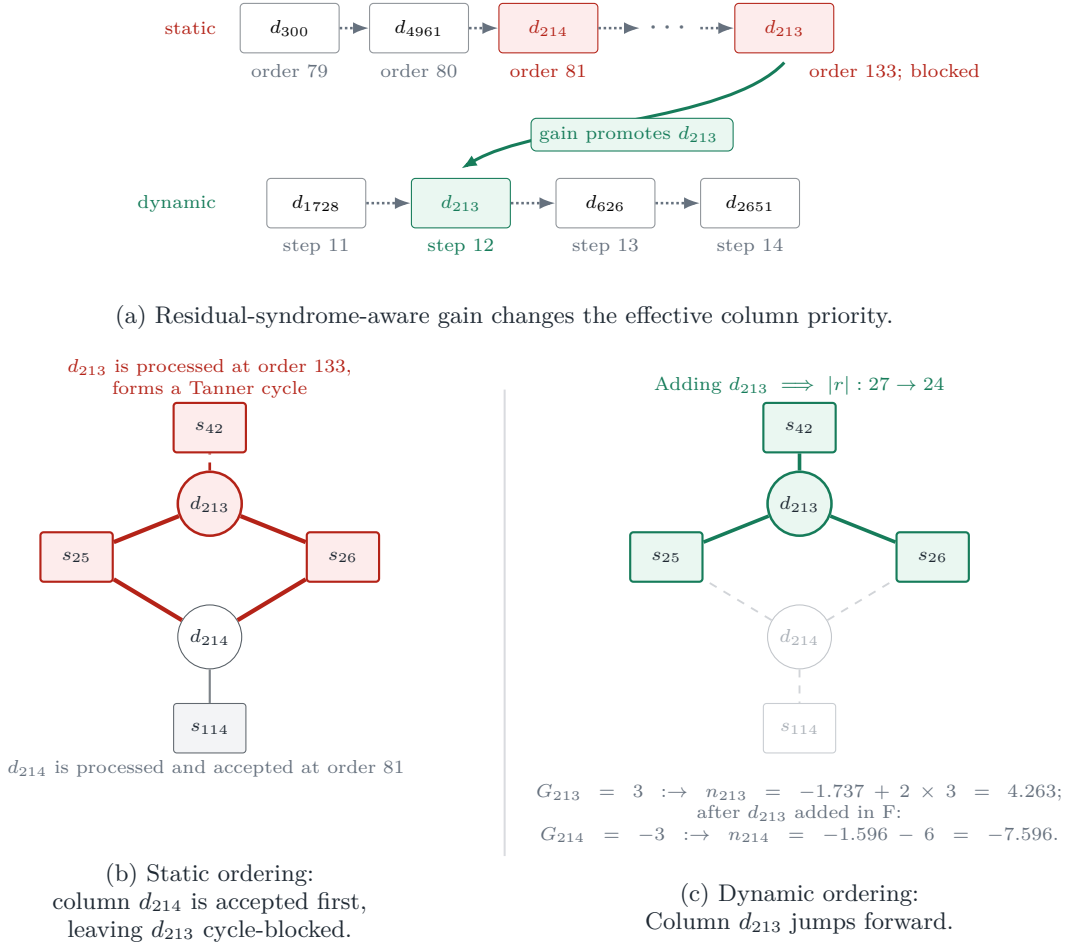
\subsection{Noisy weight assignment for BB code Decoding Matrix Columns}
\label{sec:bb-qubit-weight}
The ordered forest construction using the noise-perturbed base weight, discussed in the previous Section \ref{sec:noisy-dataqubit-weight-assignment} works tremendously for the surface codes, as we will see later from the simulation outcomes.
However, for QLDPC codes with more structures, such as the BB codes, this alone is not sufficient.
For instance, the BB code parity check matrix has weight $3$ columns, and the forests constructed according to the priority assigned by Eq. \eqref{eq:noisy-var-weight} can sometimes fail to span the syndrome.
On the other hand, for surface codes, the graph-like error mechanisms guarantee a syndrome spanning forest (for a detailed analysis see Appendix \ref{ap:syndrome-forest}).\\\\
On this note, we propose a new heuristic residual syndrome-aware online forest construction method, which results in `good' forest construction.
The base weight assignment and the noise perturbation is applied similar to what we used in the previous Section \ref{sec:noisy-dataqubit-weight-assignment}.
We assign the base weight $w_q$ to each column (data qubit) $q$, from Eq. \eqref{eq:base-var-score}.
Further, for each ensemble decoding instance of the decoder $b$, we assign a syndrome and ensemble seeded Gaussian noise $\varepsilon^{b}_q$ to the weights of the columns (data qubits) $q$.
The syndrome and ensemble instance are used as keys to construct deterministic random noise.
These deterministic seeding merely helps in debugging and benchmarking of the decoding framework and have nothing to do with the forest quality.
In fact, we can directly adapt the noise perturbation method used in the previous Section \ref{sec:noisy-dataqubit-weight-assignment}, without losing any performance gain.
For instance, in the present scenario, the seeded noise perturbation term $\varepsilon^{b}_q$ serves the same purpose of exploring different forest candidates in the ensemble, as it was accomplished using the parameter $\tau^b$ in Section \ref{sec:noisy-dataqubit-weight-assignment}.
The forest construction without these noise perturbed scores will be deterministic for all the ensemble instances and is undesirable for the greedy exploration of good inference solution choices.\\\\
We now propose a slightly different forest construction sub-routine, in-addition to what we used for the surface code.
We maintain the disjoint-set union find data structure to check if any newly added columns (data qubits) to the forest induce a cycle or not, i.e. by checking whether the neighborhood of the new column (data qubit) belongs to a distinct connected component.
However, unlike the previous construction, the column (data qubit) weights are updated continuously during the forest construction.
This crucial part is where we exploit the residual syndrome-aware reweighting.
Consider column (data qubit) $q$ as the most recent candidate for the forest that does not create a cycle.
The neighborhood of column (data qubit) $q$ is defined as $N(q) = \{c : H_{cq} = 1\}$.
We initialize the residual as $r:=s$.
Now adding column (data qubit) $q$ to the forest implies that an error supported on the column (data qubit) $q$ toggles all the neighboring syndrome values, i.e., the residual of the syndrome becomes
\begin{align}
    r \gets r \oplus H_{:q},
    \label{eq:change-residual}
\end{align}
where, $H_{:q}$ represents the $q^{\mathrm{th}}$ column of the decoding matrix $H$.
The intuition behind this argument is that a residual update like this imply that the column (data qubit) $q$ can support the violated detectors (stabilizer checks) if one of them belongs in the same set of $c \in N(q)$.
Therefore, the above Eq. \eqref{eq:change-residual} motivates us to define dynamic gain due to the addition of the column (data qubit) $q$ as
\begin{align}
    G_{q^\prime}(r) = |r| - |r \oplus H_{:q}|,
    \label{eq:gain1}
\end{align}
where the $||$ notation denotes the cardinality or, the total number of non-zero elements.
The dynamic gain $G_q(r)$ signifies the net reduction in residual syndrome weight due to the addition of column (data qubit) $q$ into the Tanner forest.
We can further simplify $G_q(r)$.
Consider $a_q(r) = \sum_{c \in N(q)}r_c$ is the total number of violated stabilizer detectors (checks) adjacent to column (data qubit) $q$, i.e. $a_q(r)$ is the total number of nonzero elements in the residual due to $q$.
Therefore, the total number of zero-elements in the residual, which will be flipped due to the addition of $q$, is $|N(q)| - a_q(r)$.
Using these two quantities, the second term of Eq. \eqref{eq:gain1} becomes
\begin{align}
    |r \oplus H_{:q}| = |r| - a_q(r) + |N(q)| - a_q(r)
\end{align}
Therefore Eq. \eqref{eq:gain1} can be written as
\begin{align}
    G_q(r) = 2a_q(r) - |N(q)|.
\end{align}
Only columns (data qubits) sharing at least one detector (check) with $q$ will have their corresponding 
$a_q(r)$ changed subsequently.
Hence their corresponding dynamic gain $G_q(r)$ will also get affected simultaneously, due to the addition of $q$ into the forest.
Therefore, for each new addition of a valid column (data qubit) $q$ into the forest, the total weight of the columns (data qubits) changes only for the set $\mathcal{A}(q) = \bigcup_{c \in N(q)}N(c)$.
The resultant updated total weight for each column (data qubit) is defined as:
\begin{align}
    n^{b}_q = \kappa w_q + \tau \varepsilon^{b}_q + \beta G_q(r).
    \label{eq:var-weight-bbcode}
\end{align}
This residual syndrome-aware forest construction, repeatedly choose the highest-weighted column (data qubit) during each new inclusion into the forest, such that no cycles form.
This greedy forest construction successfully spans the syndrome if the residual $r = 0$, otherwise it cannot be guaranteed to support the syndrome and an exact inference on the resultant forest can fail.\\\\
In Fig. \ref{fig:static-cycle-dynamic-gain-rescue}, we show a typical example, where the residual-aware forest construction driven by the dynamic gain can generate better syndrome spanning forests compared to the static case (i.e. $\beta = 0$ in Eq. \eqref{eq:var-weight-bbcode}).
In Fig. \ref{fig:static-cycle-dynamic-gain-rescue}(a), we show the effects of the dynamic gain on the column ordering compared to the static case.
The dynamic forest construction promotes the column $d_{213}$ in the processing order of the Tanner forest construction.
Whereas in the static case, the column $d_{214}$ gets processed earlier and when the construction encounters the column $d_{213}$ at later stages of the construction, it gets rejected as it forms a Tanner cycle.
Comparing the two columns $d_{213}$ and $d_{214}$; column $d_{213}$ is more important for the Tanner forest construction as it can span a self contained subset of the syndrome supported on the violated detector subset $\{s_{25}, s_{26}, s_{42}\}$.
However, $d_{214}$ can only partially support it.
The effect of adding column $d_{214}$ early instead of $d_{213}$, is the primary reason behind the non-guaranteed nature of a syndrome spanning forest construction for an arbitrary ordered processing of the decoding matrix columns (see Appendix \ref{ap:syndrome-forest} for a more detailed discussion).
Note that the term self contained subset of the syndrome represents a set of violated detector nodes, such that all of them belong to a single DEM column, which has all the adjacent detector nodes violated.
Later in Section \ref{sec:results} we show that in almost $> 98\%$ cases for the simulated BB codes, our residual syndrome-aware forest construction method is able to find syndrome spanning forests for a given circuit noise rate.
\subsection{Exact Inference on Tanner Forest using Dynamic Algorithm}
Inference on forests are exact.
An exact inference essentially solves the following optimization problem:
\begin{align}
    \arg \min_{x_q \in \mathbf{x}}\sum_{q}\mathtt{C}_qx_q : H_{F}\mathbf{x} = \mathrm{s},
    \label{eq:tree-optimization}
\end{align}
where, $H_{\mathrm{F}}$ is the bipartite decoding matrix extension for the Tanner forest $F$ and $\mathtt{C}_q$ is the cost of assigning a binary assignment $x_q = 1$ for column $q$.
The Tanner forest is constructed from the Tanner graph of the quantum code using the methods discussed in Section \ref{sec:noisy-dataqubit-weight-assignment} and \ref{sec:bb-qubit-weight}.
We now propose a two-stage dynamic program for the exact inference of the optimization problem in Eq. \eqref{eq:tree-optimization}.
The first essential requirement for the dynamic program to work is establishing a parent-child relationship between the column (data qubit) nodes and detector (check) nodes of the forest.
This assists in the upward pass of the dynamic program.
A depth-first search (DFS) algorithm \cite{tarjan1972depth} can be employed to construct such set of information.
This also converts each Tanner forest of the ensemble into a collection of rooted computation trees.
In the traditional DFS, information about new node visitations is sufficient.
However, in our work we store the information about the parent-child relationships whenever a new node is discovered.
During the DFS algorithm, whenever a child node $c$ is discovered from node $p$, we store the parent information of node $c$, i.e. $\mathtt{par}(c$) = $p$.
Simultaneously, we store the order of these new child node visitations during the DFS traversal.
Finally, for each node $p$, from the stored order, we construct a list of children nodes, i.e. $\mathtt{ch}(p) = N(p) \backslash \mathtt{par}(p)$.
Essentially, this allows an upward message passing from the leaf nodes to the root node on a directed tree component of the Tanner forest.
Therefore, the flow of messages occur in the reverse direction of the stored order, i.e. the upward message from any node in the forest is sent to its stored parent node.\\\\
Note that as the Tanner forest may contain several disjoint tree components, each connected tree component is decoded independently.
The DFS first explores one of the connected components and then traverses to another connected component.
The DFS output is stored separately for each of the disconnected components.
Also, by default, our DFS setup assumes the first non-visited column (data qubit) node as the root of the corresponding component.
We now discuss the upward pass procedure of the exact inference.
\subsubsection{The Upward Message Passing}
\label{sec:upward-pass}
The upward pass algorithm on a connected tree component essentially solves a smaller sub-tree optimization on the fly and later assists in the execution of a downward traceback of the optimal inference solution based on the decision at the root column (data qubit) node of the connected tree component.
Performing the upward pass on all the connected components of the forest optimizes Eq. \eqref{eq:tree-optimization} for the total forest $F$.\\
For the upward pass algorithm, we send a set of two entry messages $M_{c \rightarrow p}(x_p)$, from 
a child node $c$ to its parent node $p$ in the connected tree component.
We use $x_p$ to represent the bit-value assumed by a column (data qubit) node $p$.
The message set represents the minimum cost assignment of the column (data qubit) nodes of the entire sub-tree below the child node $c$, for each value assumed by a boundary column (data qubit) node, i.e. if the parent node is a column (data qubit), then for each value of the parent column (data qubit) $p$, such that $x_p \in \{0,1\}$ a message information is passed from $c$ to $p$.
Consider the case where the message is being sent from a column (data qubit) node $q$ to a parent detector (check) node $c$.
An obvious implication to the sent messages is that these messages associate a cost for assigning a value to the boundary column (data qubit) node $q : x_q \in \{0,1\}$.
These messages sent from a child column (data qubit) node $q$ to its parent detector (check) node $c$ is as follows:
\begin{align}
    M_{q \rightarrow c}(x_q) = \mathtt{C}_qx_q + \sum_{a \in \mathtt{ch}(q)}m_{a \rightarrow q}(x_q),
    \label{eq:msg_var_to_check}
\end{align}
where, $M_{q \rightarrow c}(x_q = 0)$ and $M_{q \rightarrow c}(x_q = 1)$ represents the cost of assigning the column (data qubit) values to $x_q = 0$ and $x_q = 1$ respectively.
Now consider the messages sent from a child detector (check) node to the parent column (data qubit) node.
It is defined as:
\begin{align}
    \label{eq:msg_check_to_var}
    M_{c \rightarrow q}(x_q) &= \min_{x_u : u \in \mathtt{ch}(c)}\sum_{u \in \mathtt{ch}(c)}M_{u \rightarrow c}(x_u),\\
    \text{constrained to}\,\, s_c &= x_q \oplus_{u \in \mathtt{ch}(c)} x_u.
    \notag
\end{align}
The exhaustive evaluation of Eq. \eqref{eq:msg_check_to_var} would require enumeration over $2^{|\mathtt{ch}(c)|}$ possible child column (data qubit) assignments for each value $x_q \in \{0,1\}$ of the parent column (data qubit) node $q$.
Therefore, for both the values of $x_q$, the detector (check) to column (data qubit) node messages remains a $\Theta(2^{|\mathtt{ch}(c)|})$ process.
However, we can perform a more efficient dynamic program for this binary parity-constrained optimization of Eq. \eqref{eq:msg_check_to_var}.
For any detector (check) node $c$, a two-state dynamic parity (DP) program can be employed as follows:
\begin{enumerate}
    \item The DP program essentially tracks the minimum cost child column (data qubit) assignments satisfying a required parity.
    \item We associate a dynamic parity cost with each possible parity due to the child column (data qubit) assignments.
    These parity assignment costs are updated on the fly whenever a new child column (data qubit) node assumes values from $\{0,1\}$.
    Suppose, for a particular detector (check) node $c$, the dynamic parity cost after adding the $i^{th}$ child column (data qubit) node is defined as:
    \begin{align}
    D_i(p) =
    \min_{\substack{
    x_1,\ldots,x_i \\
    x_1 \oplus \cdots \oplus x_i = p
    }}
    \sum_{j=1}^{i} M_{j \rightarrow c}(x_j),
    \label{eq:parity-dp}
    \end{align}
    where $p$ is the parity of all the processed child columns (data qubits).
    The above equation represents the cost of two possible child node parities, $p= 0$ and $p = 1$. 
    Before, processing any child nodes (i.e. $i = 0$), the dynamic parities are initialized as
    \begin{align*}
        D_0(0) = 0, D_0(1) = \infty.
    \end{align*}
    After processing the $i^{\mathrm{th}}$ child column (data qubit), the parity cost is updated as follows:
    \begin{align}
    D_i(0)
    &=
    \min\Biggl\{
    \begin{aligned}
    &D_{i-1}(0)
     + M_{i\rightarrow c}(x_i=0),\\
    &D_{i-1}(1)
     + M_{i\rightarrow c}(x_i=1)
    \end{aligned}
    \Biggr\},
    \label{eq:dp-costs-1}
    \\[2pt]
    D_i(1)
    &=
    \min\Biggl\{
    \begin{aligned}
    &D_{i-1}(0)
     + M_{i\rightarrow c}(x_i=1),\\
    &D_{i-1}(1)
     + M_{i\rightarrow c}(x_i=0)
    \end{aligned}
    \Biggr\}.
    \label{eq:dp-costs-2}
    \end{align}
    \item The target parity of all the neighbor columns (data qubits) of a particular detector (check) node $c$ is its syndrome value $s_c$.
    Therefore, the constraint on the parity of the child columns (data qubits) is $s_c \oplus x_q$, where $x_q$ is the bit-value of the parent column (data qubit) node of detector (check) $c$.
    \item The message from child $c$ to parent $q$ is the minimum cost of child column (data qubit) assignments that satisfy the required parity, i.e. 
    \begin{align*}
        M_{c \rightarrow q}(x_q) = D_{|\mathtt{ch}(c)|}[s_c \oplus x_q],
    \end{align*}
    which is already the minimum cost for the required parity assignment of the child columns (data qubits) ensured by the dynamic least cost assignment of Eq. \eqref{eq:dp-costs-1} and Eq. \eqref{eq:dp-costs-2}.
    \item While we dynamically minimize the cost of child column (data qubit) node assignments based on the syndrome bit constraint, the column (data qubit) node assignments should also be maintained on the fly, which gives the minimum cost for the dynamic parity at each step.
\end{enumerate}
This parity dynamic program process each child column (data qubit) node once and performs a constant number of $\Theta(|N(c)|)$ operations per detector (check) node.
Now, together the whole upward message pass is dominated by the number of messages passed, which is a $\Theta(|E_F|)$ process, where $E_F$ is the set of edges in the Tanner forest F.\\\\
The above dynamic parity program can further be optimized.
As we are interested in the least cost assignment of the columns (data qubits) from the Tanner forest, the two-state dynamic parity program can be significantly optimized into a cheapest column (data qubit) assignment using the cheapest parity flip rule-set.
The exact process is described in Algorithm \ref{alg:parity-dp}, we name it the $\mathtt{CheckUpdate}$ function.
It represents the exact optimization achieved by the dynamic parity program discussed above.
A more detailed argument and discussion for the algorithm $\mathtt{CheckUpdate}$ can be found at Appendix \ref{ap:parity-dp}.
\SetAlgoNoLine
\DontPrintSemicolon
\setlength{\algomargin}{0.32em}
\SetInd{0.22em}{0.36em}
\SetKwInOut{KwIn}{Input}
\SetKwInOut{KwOut}{Output}
\SetKwFunction{CheckUpdate}{CheckUpdate}
\SetKwProg{Fn}{Function}{:}{}
\begin{algorithm}[t]
\caption{Dynamic Parity Program at\\ detector (check) node $c$}
\label{alg:parity-dp}

\Fn{\CheckUpdate{}}{

\KwIn{Detector (check) node $c$, syndrome bit $s_c$, and message from child column (data) to parent detector (check) node: $\{M_{q \rightarrow c}(x_q = 0),M_{q \rightarrow c}(x_q = 1)\}_{\forall q\in\mathtt{ch}(c)}$}

\KwOut{Message sent from detector (check) to parent column (data) node: $M_{c \rightarrow p}(x_p)$ and stored child data assignments
$\mathcal{D}_{c \rightarrow p}$}

\For{$x_p\in\{0,1\}$}{
    $w\gets s_c\oplus x_p$\;

    \ForEach{$q\in\mathtt{ch}(c)$}{
        $x^{\star}_q\gets
        \arg \min\limits_{x_q\in\{0,1\}}M_{q \rightarrow c}(x_q)$\;

        $\Delta_q\gets
        M_{q \rightarrow c}(x_q = 1-x^{\star}_q)-M_{q \rightarrow c}(x_q = x^{\star}_q)$\;
    }

    \If{$\displaystyle
        \bigoplus_{q\in\mathtt{ch}(c)}x^{\star}_q\neq w$}{

        $q^\star\gets
        \displaystyle\arg\min_{q\in\mathtt{ch}(c)}
        \Delta_q$\;


        $x^{\star}_{q^\star}\gets x^{\star}_{q^\star}\oplus1$\;
    }

    $M_{c \rightarrow p}(x_p)\gets
    \displaystyle\sum_{q\in\mathtt{ch}(c)}
    M_{q \rightarrow c}(x_q = x^{\star}_q)$\;

    $\mathcal{D}_{c \rightarrow p}(x_p)\gets
    \{x^{\star}_q : q \in \mathtt{ch}(c)\}$\;
}

\Return{$(M_{c \rightarrow p}(x_p), \mathcal{D}_{c \rightarrow p})$}\;
}

\end{algorithm}

\subsubsection{The Downward Traceback procedure}
The upward message pass is followed by a downward traceback to complete the exact inference algorithm.
It traces all the optimum column (data qubit) assignments of every connected tree component of the forest.
Our dynamic program supports efficient reconstruction of the optimal column (data qubit) node assignments for the cheapest solution to the inference problem.
During the upward message passing stage, we store the constrained minimum cost child column (data qubit) assignments for each detector (check) node $c$.
For any value that the parent column (data qubit) node $p$ assumes (i.e. $x_p \in \{0, 1\}$), we store
\begin{align}
\mathcal{D}_{c \rightarrow p}(x_p) = \{x^{\star}_{q} : q \in \mathtt{ch}(c)\}
\end{align}
from Algorithm \ref{alg:parity-dp}.
For each connected tree component of the Tanner forest, the downward traceback starts from a decision stage at the root column (data qubit) node of the component.
The root node $r$ associates a total cost of
\begin{align}
    M_r(x_r) = \mathtt{C}_rx_r + \sum_{c \in \mathtt{ch}(r)}M_{c \rightarrow r}(x_r), \forall x_r \in \{0, 1\},
\end{align}
where $\mathtt{C}_r$ is the local cost of assigning $x_r = 1$, and it is obtained from the synthetic soft information discussed previously in Section \ref{sec:noisy-dataqubit-weight-assignment} and \ref{sec:bb-qubit-weight}.
The root column (data qubit) node is assigned with the value having the minimum accumulated cost, i.e. 
\begin{align}
    x^{\star}_r = \arg \min_{x_r \in \{0, 1\}}M_r(x_r).
\end{align}
Once the optimal root decision has been made, i.e. $x_r \gets x^{\star}_r$, the traceback starts retrieving the stored child column (data qubit) assignments $x^{\star}_q$ for each detector (check) node given the assigned value of the parent column (data qubit) node $x^{\star}_p$, i.e.
\begin{align} 
\left\{ x^{\star}_q: q\in \mathtt{ch}(c) \right\} \gets \mathcal{D}_{c\rightarrow p}(x^{\star}_p). 
\end{align}
The complete exact inference algorithm is described in Algorithm \ref{alg:exact-map}.
The traceback is also a $\Theta(|E_F|)$ process.
The whole exact inference algorithm incurs a cost of $\Theta(|E_F|)$, which is a linear one.
\SetAlgoNoLine
\DontPrintSemicolon
\setlength{\algomargin}{0.35em}
\SetInd{0.25em}{0.40em}
\SetKwInOut{KwIn}{Input}
\SetKwInOut{KwOut}{Output}
\SetKwFunction{CheckUpdate}{CheckUpdate}
\SetKwFunction{DFS}{DFS}
\SetKwFunction{ETFI}{ETFI}
\SetKwProg{Fn}{Function}{:}{}
\begin{algorithm}
\caption{Exact Tanner Forest Inference}
\label{alg:exact-map}

\Fn{\ETFI{}}{

\KwIn{Tanner forest $F=(V_F \cup C_F, E_F)$, syndrome
$\mathbf{s}$, and synthetic column (data qubit) LLRs $\boldsymbol{\lambda}^{\mathrm{noise}}$}

\KwOut{Inference solution with least cost $\hat{\mathbf{e}}$}
\BlankLine
\ForEach{nontrivial component $T$ of $F$}{
    \DFS(T)\;

    \ForEach{$u$ in reverse postorder}{
        \uIf{$u$ is a non-root column (data qubit) node}{
            \For{$x_u\in\{0,1\}$}{
                $M_{u \rightarrow \mathtt{par}(u)}(x_u)\gets
                -\lambda^{\mathrm{noise}}_u x_u+
                \displaystyle\sum_{c\in\mathtt{ch}(u)}M_{c \rightarrow u}(x_u)$\;
            }
        }
        \ElseIf{$u$ is a detector (check) node}{
            $\begin{aligned}
            &(M_{u \rightarrow \mathtt{par}(u)},\mathcal{D}_{u \rightarrow \mathtt{par}(u)})
            \\[-0.2em]
            &\qquad\gets
            \CheckUpdate\Bigl(
                u,s_u,
                \\[-0.2em]
            &\qquad\qquad
                \{M_{q \rightarrow u}(x_q=0),
                  M_{q \rightarrow u}(x_q=1)\}
                _{q\in\mathtt{ch}(u)}
            \Bigr)
            \end{aligned}$\;
        }
        \uIf{$u$ is a root column (data) node}{
        \For{$x_u \in \{0,1\}$}{
            $M_u(x_u) = -\lambda^{\mathrm{noise}}_ux_u + \sum\limits_{c \in \mathtt{ch}(u)}M_{c \rightarrow u}(x_u)$
            }
            $x^{\star}_u\gets\displaystyle\arg\min_{x_u\in\{0,1\}}M_u(x_u)$\;
        }
    }

    Traverse $T$ downward from root column node $r$, with $x_r = x^{\star}_r$\;

    At any non-root column node $q$, if
    $x^{\star}_q \gets b$, pass $b$ to its child detector node.\;

    If detector (check) $c$ gets value $b$ from its parent $q$, assign its child column (data) nodes according to
    $\mathcal{D}_{c \rightarrow q}(x_q = b)$.\;
}

\ForEach{isolated column (data) node $q$ of $F$}{
    $x^{\star}_q\gets\mathbb{I}[\lambda^{\mathrm{noise}}_q<0]$\;
}

\Return{$\hat{\mathbf{e}} \gets 
\mathbf{x}^{\star}$}\;
}
\end{algorithm}

\subsubsection{Ensemble decoding}
Ensembling is very efficient for strategizing accurate decoding outcomes when individual decoding attempts might have shortcomings \cite{shutty2026efficient}.
The key intuition of this technique is a parallel implementation of multiple decoding instances and net pool a decoding decision from this ensemble.
For instance, inference carried out on a Tanner forest is exact.
However, it outputs a MAP solution, which in general might not be optimal for the DQMLD.
Further, we construct the Tanner forests from the synthetic soft information, derived from the channel LLRs and syndrome information with added noise.
These synthetic soft informations do not guarantee a syndrome spanning forest for the general-purpose QLDPC codes (Appendix \ref{ap:syndrome-forest}).
This implies that, despite removing cycles and mitigating degenerate trapping set-like configurations, it is highly non-trivial to find a set of edges that can be removed with confidence during the forest construction, such that a valid logical correction supporting the syndrome can be found.
Therefore, choosing an optimum removal set for the data qubits during the forest construction is impossible beforehand, either with the noise-assisted data qubit scoring or even with the BP soft LLR guided scoring.  
Therefore, creating an ensemble of Tanner forests helps mitigate this issue.
From the Tanner graph, we construct an ensemble of forests and execute Algorithm \ref{alg:exact-map} on each of those ensembles.\\\\
In our work, we chose different decoding instances with varying noise strength as an ensemble choice.
In any instance, the data qubit weights are incorporated with an i.i.d Gaussian noise.
We either scale the strength of these noises with a fixed scaling parameter $\tau^b$ for an ensemble decoding instance $b$ or perturb the base scores with ensemble-seeded noise itself. 
The final decoding outcome can be obtained from the results of different parallel decoders using a pooling strategy.
We explore two such strategies.
One of them is the most-likely-error (MLE) pooling strategy reported in Ref. \cite{shutty2026efficient}.
Consider an ensemble of $B$ decoders.
The MLE pooling returns the most likely outcome from the ensemble.
Consider $\mathtt{C}_q$ as the cost associated with the qubit $q$ suffering an error.
The total cost of an predicted error $\mathbf{e}^b$ from one of the ensemble decoding instance $b$ is:
\begin{align}
    \mathbf{C}^b(\hat{\mathbf{e}}^{b}) = \sum_{q=1}^{n}\mathtt{C}_q e^{b}_q.
\end{align}
The pooled outcome $\hat{\mathbf{e}}$ from the ensemble is obtained follows:
\begin{align}
    \hat{\mathbf{e}} := \arg \min_{ \hat{\mathbf{e}}^{b}}\mathbf{C}^b(\hat{\mathbf{e}}^{b}).
\end{align}
The other method we use is the first come first serve method.
This method essentially outputs the prediction from the ensemble instance that produces a syndrome supporting solution first.
This method, unlike the MLE pooling does not require all the ensemble instances to finish their decoding. 
We now perform circuit-level simulations to test the error correcting capabilities the NAED decoder and also benchmark its decoding speed compared to BP+OSD$0$.
\begin{figure*}
    \centering
    \includegraphics[width=0.65\textwidth]{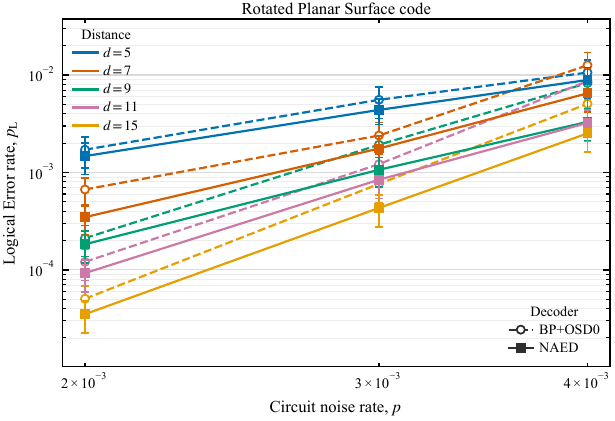}
    \caption{We show the circuit-level logical error rates of the rotated planar surface codes obtained under both the NAED and BP+OSD$0$ decoder.
    The NAED contains an ensemble of $11$ decoding instances with the noise scaling factor uniformly distributed across the $11$ forest constructions in the range $\{0, 0.5\}$.
    Also, the synthetic soft information used for the Tanner forest constructions is used as in Eq. \eqref{eq:noisy-var-weight}.
    The BP iterations are capped at a maximum of $1000$ for the BP+OSD$0$ decoder.
    The circuit-level noise model assumed for the simulations are described in the Appendix \ref{ap:circuit-simulation}.
    A logical error occurs if the decoder-predicted logical observable does not match the logical of the prepared state.
    The LERs are obtained at a $95\%$ confidence interval.
    }
    \label{fig:surface-ler-comparison}
\end{figure*}

\SetAlgoNoLine
\DontPrintSemicolon
\setlength{\algomargin}{0.6em}
\SetInd{0.35em}{0.55em}
\SetKwFunction{Forest}{Forest}
\SetKwFunction{ETFI}{ETFI}
\SetKwProg{Fn}{Function}{:}{}
\begin{algorithm}[H]
\caption{A general NAED framework}
\label{alg:naed}

\KwIn{Detector Error Model matrix: $H$, syndrome: $\mathbf{s}$, channel LLRs: $\boldsymbol{\mathtt{llr}}$,
ensemble size $B$, ensemble noise scales: $\{\tau^b\}_{\forall b \in B}$}
\KwOut{Decoder prediction $\hat{\mathbf{e}}$}

Compute a syndrome and channel LLR dependent base score (Eq. \eqref{eq:base-var-score}): $\boldsymbol{\omega}=\Phi(H,\mathbf{s},\boldsymbol{\mathtt{llr}})$\;

\ForPar{$b=1,\ldots,B$}{
    Build a static or syndrome-aware dynamic Tanner forest from full or sparse DEM $H$.\;
    $F^b, H^{b}_{F}$ = \Forest($H, \boldsymbol{n}^{b} = \eta(\boldsymbol{\omega}, \tau^b)$)\;

    Perform exact forest inference (Algorithm \ref{alg:exact-map}):\;
    $\hat{\mathbf{e}}^{b} \gets$ \ETFI($F^b$, $\mathbf{s}$, $\mathbf{n}^b$)\;
    
    \If{$\mathbf{s}\notin\operatorname{span}(H^{b}_{F})$}{
        Discard $F^b$ and continue\;
    }

    \If{DEM matrix $H$ is sparse}{
        Lift $\hat{\mathbf{e}}^{b}$ to the full DEM data qubit set\;
    }

    \If{$H\hat{\mathbf{e}}^{b}=\mathbf{s}$}{
        store $\hat{\mathbf{e}}^{b}$ with cost
        $\mathtt{z}^b=\boldsymbol{\mathbf{C}^{b}}(\hat{\mathbf{e}}^{b})$\;
    }
    \If{$\mathtt{early\_stop} = \mathrm{True}$}{
        \Return{\textnormal{first valid} $\hat{\mathbf{e}}^{b}$}\;
    }
}

\Return{$\arg\min\limits_{\hat{\mathbf{e}}^{b}} \mathtt{z}^b$}\;

\end{algorithm}

\section{Simulation Results}
\label{sec:results}
We now perform circuit-level simulations of the rotated planar surface code and the BB codes with our NAED decoder.
The circuit-level simulations assume a more general noise model, where the stabilizer measurement circuits are themselves faulty.
Therefore, for the decoding purposes, the stabilizer measurements themselves can not be decisive and to mitigate faulty stabilizer measurements a repeated number of error correction rounds needs to be performed.
Under these circumstances, the syndromes obtained from the noisy stabilizer circuits are recovered by the decoder using the detector error model (DEM) \cite{gidney2021stim, derks2025designing}.
DEMs can capture a vast set of error mechanisms through sets of deterministic linear combinations of circuit measurements, called the detectors.
Accordingly, under circuit-level noise, decoding is performed on the DEM rather than on the conventional parity check matrix.
The rows of the DEMs correspond to detectors, while the columns represent circuit-level error mechanisms rather than individual data qubit errors.
The decoder, therefore, takes the observed detector events as its input syndrome, rather than the outcomes of individual stabilizer measurements.\\\\
However, these DEMs are generally not sparse, whereas the parity check matrix of the underlying QLDPC codes is sparse.
NAED, like many other efficient QLDPC decoders, works best on a sparse decoding matrix.
The lack of DEM sparsity raises a unique problem in terms of achieving a syndrome spanning forest from the underlying DEM decoding matrix.
In general, a syndrome spanning forest cannot be guaranteed from the DEM matrix (see more in Appendix \ref{ap:syndrome-forest}).
In Ref. \cite{demarti2026almost} the authors propose a sparsification routine, which constructs a sparse DEM matrix with bounded column weights.
We use the same sparsification procedures from Ref. \cite{demarti2026almost} for the surface and BB codes.
Essentially, the sparsification tries to capture the sparse nature of the underlying code's parity check matrix into the circuit DEM.
We show in Appendix \ref{ap:syndrome-forest} that sparse DEMs still can not guarantee a syndrome spanning forest unless the sparse DEM has only graph-like error mechanisms.
This is also argued in the Ref. \cite{demarti2026almost}.
For surface codes, the sparse DEM matrix contains graph-like columns and therefore a syndrome spanning forest can be guaranteed regardless of the inclusion order of the decoding matrix columns into the forest.\\\\
We simulate various distance $d$ rotated planar surface codes using the NAED framework and further compare the decoding performances with the BP+OSD$0$ decoder.
We create an ensemble of $B = 11$ decoding instances and set a uniform noise scaling in the range of $\tau \in \{0, 0.5\}$ for each of the $11$ ensemble decoders.
In Fig. \ref{fig:surface-ler-comparison} we show the logical error rates (LER) obtained from both the NAED and BP+OSD$0$ decoder.
The LERs are obtained with a $95\%$ confidence interval, indicated using the error bars.
A detailed setup for the circuit level simulations is discussed in Appendix \ref{ap:circuit-simulation}.
We observe that for surface code, the Tanner forest construction using noise perturbed synthetic soft informations (Eq. \eqref{eq:noisy-var-weight}) is sufficient to achieve OSD-like LERs.
In all the simulated distances, NAED can outperform BP+OSD$0$ with just $11$ parallel decoding instances in the ensemble.\\
We also estimate a comparison of per error-correcting round decoding times for each surface code distance.
We compare the decoding time of BP+OSD$0$ with NAED in Fig. \ref{fig:benchmark-surface}.
For the NAED, we use $10$ parallel decoding instances, each executed on parallel CPU workers of an M$4$ MacBook chip.
The decoding time of the NAED ensemble is noted as soon as any of the workers report a syndrome valid output, i.e., the decoding time for the ensemble decoders is determined by the time required to get the first syndrome valid solution.
For BP+OSD$0$, we set a fixed cap of $1000$ BP iterations.
From Fig. \ref{fig:benchmark-surface} we observe that the NAED decoding time outperforms the BP+OSD$0$ in all the simulated surface code distances, with NAED outperforming BP+OSD$0$ by a couple of orders as more resource-intensive decoding is required for the higher distance codes.
\begin{figure}
    \centering
    \includegraphics[width=0.5\textwidth]{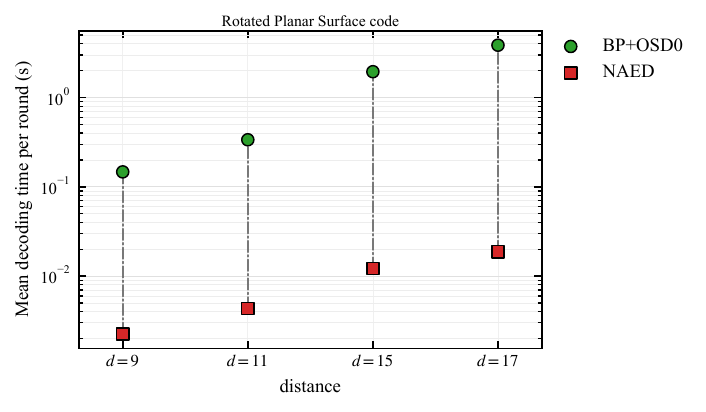}
    \caption{We compare the mean per error correction round decoding time of the NAED and BP+OSD$0$ for various distance surface codes.
    The NAED decoder follows Algorithm \ref{alg:naed}.
    The NAED ensemble contains $10$ decoding instances working in parallel on an M$4$ MacBook chip.
    We report the NAED decoding time from the initialization to the time of the first obtained syndrome valid solution from any of the parallel workers.
    We sample $1000$ syndromes at a fixed circuit noise rate $p = 0.002$ and obtain the mean decoding times per error correction round.
    }
    \label{fig:benchmark-surface}
\end{figure}
%
\begin{table}[t]
    \centering
    \caption{Success rate comparisons of a multi stage decoding setup. Code: Surface code, min-sum BP normalization factor $\alpha = 1.0$, circuit noise rate $p=0.002$.}
    \label{tab:surf-solution-fraction-1}
    \renewcommand{\arraystretch}{1.2}
    \setlength{\tabcolsep}{2pt}
    \begin{tabular}{| l |c |}
        \toprule
        \textbf{Decoding stage}
        &
        \textbf{Success percentage}
        \\
        \midrule
        \hline
        BP                 & $43.7\%$  \\
        Forest-based  & $56.3\%$ \\
        \bottomrule
    \end{tabular}
\end{table}
\begin{table}[t]
    \centering
    \caption{Success rate comparisons of a multi stage decoding setup. Code: BB code, min-sum BP normalization factor $\alpha = 1.0$, circuit noise rate $p=0.002$.}
    \label{tab:bb-solution-fraction-1}
    \renewcommand{\arraystretch}{1.2}
    \setlength{\tabcolsep}{2pt}
    \begin{tabular}{| l |c |}
        \toprule
        \textbf{Decoding stage}
        &
        \textbf{Success Percentage}
        \\
        \midrule
        \hline
        BP                & $99.7\%$ \\
        Forest-based       & $0.3\%$  \\
        \bottomrule
    \end{tabular}
\end{table}
%
\begin{table}[ht!]
    \centering
    \caption{Success rate comparisons of a multi stage decoding setup. Code: Surface code, min-sum BP normalization factor $\alpha = 0.625$, circuit noise rate $p=0.002$.}
    \label{tab:surf-solution-fraction-625}
    \renewcommand{\arraystretch}{1.2}
    \setlength{\tabcolsep}{2pt}
    \begin{tabular}{| l |c |c |}
        \toprule
        \textbf{Decoding stage}
        &
        \textbf{Success Percentage}
        \\
        \midrule
        \hline
        BP                  & $0.7\%$  \\
        Forest-based  & $99.3\%$ \\
        \bottomrule
    \end{tabular}
\end{table}
\begin{table}[ht!]
    \centering
    \caption{Success rate comparisons of a multi stage decoding setup. Code: BB code, min-sum BP normalization factor $\alpha = 0.625$, circuit noise rate $p=0.002$.}
    \label{tab:bb-solution-fraction-625}
    \renewcommand{\arraystretch}{1.2}
    \setlength{\tabcolsep}{2pt}
    \begin{tabular}{| l |c |}
        \toprule
        \textbf{Decoding stage}
        &
        \textbf{Success Percentage}
        \\
        \midrule
        \hline
        BP                & $73.5\%$ \\
        Forest-based      & $24.2\%$  \\
        \bottomrule
    \end{tabular}
\end{table}

\noindent Before describing the experiments with the BB codes, we study another important aspect of decoding over a Tanner forest.
We first study the syndrome resolving capabilities of BP on both the full and sparsified DEMs for the distance $15$ rotated planar surface code and the $[[144, 12, 12]]$ BB code.
Further, we report the number of unresolved syndromes that is properly addressed by decoding over a forest and not by the previous BP decoding stage.
For a standard circuit-level noise model (Appendix \ref{ap:circuit-simulation}) and constant circuit noise rate parameter $p = 0.002$, we generate $10$K syndrome samples.
We create an ensemble of BP decoders with a varied number of total message passing iterations, in steps of $17$ from the range $\{17, 391\}$.
The ensemble consists of $23$ decoding instances.
If any of the decoding instance from the ensemble predicts a syndrome supporting solution, success is reported.
We first apply BP on the full DEM.
The syndromes that are not corrected by the former round are fed into a round of BP on the sparse DEM, with the original channel priors mapped into the sparse channel priors.
Syndromes that are not resolved by these two BP stages are then fed into a forest inference stage using Algorithm \ref{alg:exact-map}.
In this final round, we construct an ensemble of $23$ Tanner forests for each sampled syndrome that are not solved using the two prior BP stages.
We construct these Tanner forests from the full DEM and later from the sparse DEMs, using the noise-perturbed methods discussed previously.
A solution is reported if Algorithm \ref{alg:exact-map} succeeds in finding a syndrome-supported solution from any of these inference stages on the ensemble of Tanner forests.\\\\
We summarize a report for the success ratios of the two primary decoding stages (i.e., BP-based and forest-based); in Table \ref{tab:surf-solution-fraction-1}, \ref{tab:bb-solution-fraction-1}, \ref{tab:surf-solution-fraction-625}, \ref{tab:bb-solution-fraction-625}.
We studied the impact on BP stage success due to the two normalization factors $\alpha = 1.0$ and $\alpha = 0.625$ for the min-sum algorithm of BP.
The results indicate that surface code can be critical to BP decoding.
We observe that BP is especially detrimental over the full DEM for surface codes.
However, its graph-like error mechanisms make it enticing for decoding over the sparse DEM, which dominantly contributes to the BP stage success rate at $\alpha = 1.0$.
Further, a significant amount of syndrome samples still cannot be resolved by either stages of BP and the exact forest inference of Algorithm \ref{alg:exact-map} approximately solves all the unresolved syndromes from the BP stages.
The $99\%$ success rate of Algorithm \ref{alg:exact-map} from Table \ref{tab:surf-solution-fraction-625} shows the tremendous advantage of a forest-based inference for surface codes.\\
On the other hand, for the BB code, we observe that the two BP stages are very much capable of finding a solution to the syndrome, with the success ranging from $70-99\%$.
Whereas, the forest-based inference has been called only for a small number of instances, because a significant number of sampled syndromes are already addressed by the two BP decoding stages for BB codes.
This particularly implies that in general BB code might incur structures that are much more favorable to BP, even so, with the right normalization factor, and there also exist cases where the errors cannot be supported by the forests.
For instance, a logically valid solution resembling a Tanner cycle can never be spanned by a Tanner forest, and therefore, the exact inference algorithm will never be able to predict that valid correction from the forest itself.
Therefore, for surface code, the complete failure of BP from Table \ref{tab:surf-solution-fraction-625}, implies forest based inference can show significant decoding improvements under the sparse DEMs.
However, for BB codes, there are certain error configurations that cannot be solved by forest-based decoding.
This particular observation lead us to try a low-cost post-processing with our NAED framework to target specifically these unsolavble errors.\\\\
We employ the following decoding pipeline for the BB codes, for any given syndrome:
\begin{itemize}
    \item First we employ the NAED over the full DEM of the BB code.
    \item If the first stage fails to resolve the syndrome, we employ NAED over the sparse DEM.
    \item If both the NAED stages fail, we run BP iterations over the full DEM and later on the sparse DEM.
\end{itemize}
\begin{figure*}[t]
    \centering

    \begin{minipage}[t]{0.48\textwidth}
        \centering
        \includegraphics[width=\linewidth]{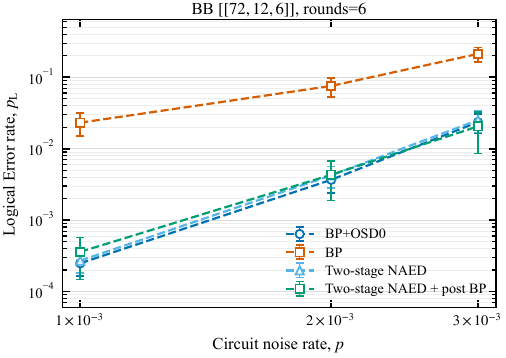}
        \par\smallskip
        \textbf{(a)}
    \end{minipage}
    \hfill
    \begin{minipage}[t]{0.48\textwidth}
        \centering
        \includegraphics[width=\linewidth]{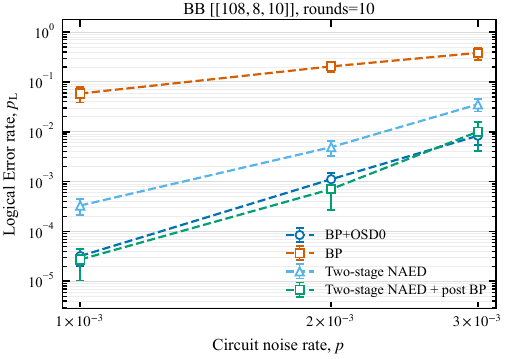}
        \par\smallskip
        \textbf{(b)}
    \end{minipage}

    \vspace{0.8em}

    \begin{minipage}[t]{0.50\textwidth}
        \centering
        \includegraphics[width=\linewidth]{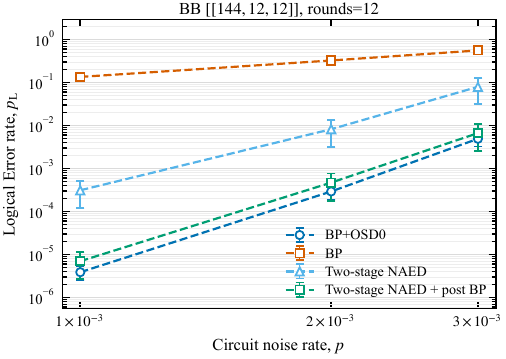}
        \par\smallskip
        \textbf{(c)}
    \end{minipage}

    \caption{We show the circuit-level LERs of BB codes obtained under various decoding methods.
    We compare the standard min-sum BP variant and BP+OSD$0$ with the NAED variants proposed for general-purpose QLDPC codes.
    The `two-stage NAED' first tries to solve the syndrome using Algorithm \ref{alg:naed} with the full DEM as input.
    If the syndrome is not resolved using the full DEM, Algorithm \ref{alg:naed} is executed on the sparse BB DEM.
    The `Two-stage NAED + post BP' is motivated by the findings from Table \ref{tab:bb-decoding-stage-breakdown}.
    This NAED variant invokes a low-cost post-min-sum BP algorithm over the full and sparse DEMs if the two NAED stages cannot resolve the input syndrome.
    The ensemble size of the NAED is $B = 100$.
    The synthetic soft information for the BB code DEMs follows from Eq. \eqref{eq:var-weight-bbcode}.
    The scalar parameters used for the DEM column weights are as follows: $\kappa = 0.5, \tau = 0.75, \beta = 2.0$.
    The ensemble forest diversity is introduced by an ensemble instance seeded noise $\varepsilon^{b}_q$.
    }
    \label{fig:bb-ler-comparison}
\end{figure*}
\begin{table*}[t!]
    \centering
    \caption{Success rate comparisons of a $4$ stage NAED + post-BP decoding setup.
    Each entry reports the number of solved shots and the corresponding success percentage; from a total of $10$K sampled syndromes. circuit noise rate $p = 0.002$.}
    \label{tab:bb-decoding-stage-breakdown}

    \renewcommand{\arraystretch}{1.2}
    \setlength{\tabcolsep}{2pt}

    \begin{tabular}{| l | c | c | c | c | c}
        \toprule
        \textbf{Code}
        &
        \makecell{\textbf{Forest (full DEM)}\\
                  \textbf{Success percentage}}
        &
        \makecell{\textbf{Forest (Sparse DEM)}\\
                  \textbf{Success percentage}}
        &
        \makecell{\textbf{Post-BP (full DEM)}\\
                  \textbf{Success percentage}}
        &
        \makecell{\textbf{Post-BP (Sparse DEM)}\\
                  \textbf{Success percentage}}
        \\
        \midrule
        \hline
        BB $[[108,8,10]]$
        & $9{,}911 \,/\, 10{,}000: 99.11\%$
        & $43 \,/\, 10{,}000: 0.43\%$
        & $11 \,/\, 10{,}000: 0.11\%$
        & $5 \,/\, 10{,}000: 0.05\%$
        \\

        BB $[[144,12,12]]$
        & $9{,}845 \,/\, 10{,}000: 98.45\%$
        & $84 \,/\, 10{,}000: 0.84\%$
        & $17 \,/\, 10{,}000: 0.17\%$
        & $14 \,/\, 10{,}000: 0.14\%$
        \\

        \bottomrule
    \end{tabular}
\end{table*}
In the above 4-stage setup, the channel LLRs and the synthetic soft information are mapped to their sparse DEM equivalent, while the decoding stage happens over the sparse DEM.
In table \ref{tab:bb-decoding-stage-breakdown}, we report the solution finding capabilities of these different stages for the $[[108,8,10]]$ and $[[144,12,12]]$ BB code.
We again sample $10$K syndromes for a fixed circuit noise rate $p = 0.002$ under the same circuit level noise.
We observe that Algorithm \ref{alg:exact-map} can solve more than $98\%$ syndromes from the full DEM.
Whereas the BP stages contribute only $0.1 - 0.2 \%$ syndrome resolutions.
However, this small resolution is necessary to reach OSD accuracy for larger BB codes.
For the circuit-level memory experiments of the BB codes, we use around $B = 100$ parallel decoding instances in the NAED ensemble.
For the post-BP rounds, the ensemble sets the number of message passing iterations uniformly distributed in the range $\{17, 200\}$ over the $100$ parallel decoding instances.
We show the LER comparisons for different BB codes in Fig. \ref{fig:bb-ler-comparison}.
We observe that for higher BB codes, NAED alone is unable to achieve OSD accuracy due to errors that cannot be supported by the forest, and therefore Algorithm \ref{alg:exact-map} cannot predict it.
Under those circumstances, employing a low-cost post-BP ensemble helps close the decoding gap.
We also compare the per error correcting round decoding time benchmarks for different BB codes in Fig. \ref{fig:benchmark-bb}.
We observe that the forest-based ensemble NAED can achieve a couple of orders of improvement in the decoding time compared to the BP+OSD$0$.

\begin{figure}
    \centering
    \includegraphics[width=0.5\textwidth]{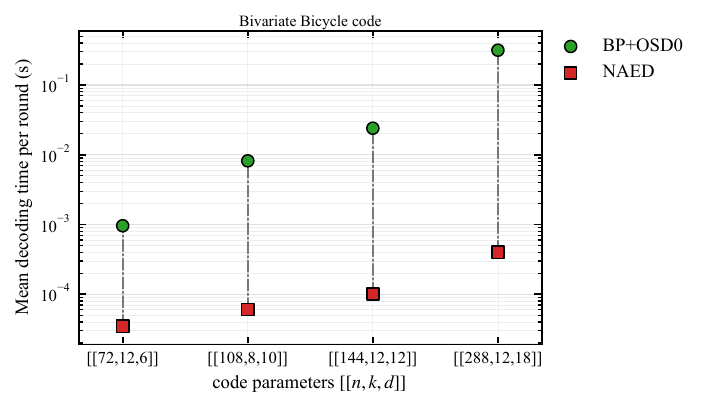}
    \caption{We compare the mean per error correction round decoding time of the NAED and BP+OSD$0$ for various BB codes.
    In this benchmark, we compare the decoding time of the two-stage NAED, i.e., Algorithm \ref{alg:naed} executed with the full and sparse DEM.
    The two-stage NAED refers to an ensemble of $10$ decoding instances, with each instance running NAED on the full DEM, and if the input syndrome is not resolved by this stage, the syndrome is fed on to a NAED on the sparse DEM.
    At a circuit noise rate of $p = 0.002$, we sample $1000$ syndromes and compare the decoding times.
    The two stage NAED decoding time notes the time till any of the ensemble instance output a syndrome valid prediction.}
    \label{fig:benchmark-bb}
\end{figure}

\section{Conclusion and Future Works}
We propose a Noise-Assisted Ensemble Decoding that is primarily based on the fact that inference on a forest is always exact.
Given a syndrome, inference on a syndrome-spanning forest can always produce an outcome that resolves the syndrome.
We propose a dynamic program algorithm which uses a single round of message passing and assignment traceback to infer the solution from the Tanner forest.
We propose the use of noise-perturbed synthetic soft information in both the forest construction and the dynamic program for inference.
Our analysis shows that the synthetic soft information contains sufficient information, and that the noise perturbation provides enough ensemble diversity to surpass the decoding performance achieved by BP+OSD$0$ for the surface codes.
The successful outcomes of using synthetic soft information instead of the BP soft information reflect over a key aspect of the standard message passing rules of BP.
It implies that although BP-derived soft LLRs are widely used to capture information about complex error propagation in noisy stabilizer circuits, their reliability can deteriorate for particularly harmful error configurations, such as quantum trapping sets. In these regimes, standard BP soft information may fail to distinguish the relevant error support, and the message-passing dynamics can oscillate or remain non-convergent. Consequently, imposing a fixed cap on the number of BP iterations may either terminate before a valid solution is reached or incur excessive decoding latency, highlighting an inherent limitation of relying on heuristic iterative message passing for real-time decoding.
Whereas, our dynamic program gives a single-shot inference.\\\\
Our motivation for designing the synthetic soft information is to preferentially incorporate vulnerable columns of the decoding matrix into the Tanner forest construction. 
By combining syndrome-aware decoding matrix column weighting with ensemble noise perturbations, a set of informative soft priorities are generated.
It helps identify the neighborhoods of the decoding matrix columns that are most likely to support the observed syndrome. 
This proposed weighting mechanism, conceptually and in terms of performance, resembles the soft information produced by unsuccessful BP iterations, on the premise that BP fails to recover the exact error configuration in the presence of detrimental error supports.
We show that these proposed column weights can guide a good Tanner forest construction, thereby enabling accurate inference outcomes.
For general-purpose QLDPC codes, such as the BB codes, we argue that the forest inference can be of tremendous advantage once the Tanner forest spans the syndrome.
However, in general, a syndrome spanning Tanner forest cannot be guaranteed, and we propose a multistage NAED decoding with a low-cost post-BP stage in such scenarios.
We also show that even in the multi-stage NAED with low-cost post-processing, the forest-based inference can correct a significant portion of the errors.
The limitation of the forest-based inference arises due to the fact that some errors can never be supported in a forest.
Therefore, in general, there cannot be an optimum set of soft information that can guarantee a syndrome-spanning Tanner forest.\\\\
In future works, we will explore the hardware implementation of NAED on FPGAs.
Our benchmarks have already shown NAED as a well-suited candidate for real-time decoding for resource-intensive circuit-level fault resolution.
Therefore, practical insights into the implementation and algorithmic evolution of the dynamic program for the inference algorithm and syndrome-aware forest constructions using synthetic soft information are of good interest.
Further, the ensemble-decoding nature of NAED makes it highly demanding for implementation on hardware designed specifically to take advantage of high parallel computing power, such as GPUs or FPGAs.
We expect an attractive performance in throughput and latency with the real-time implementations of the NAED decoder.\\\\
We also plan to study whether the synthetic soft information computations can be further optimized for general-purpose QLDPC codes.
As we move towards higher distance codes, the fraction of syndrome-spanning forests that can be obtained using any soft-information-guided construction decreases.
Therefore, we require larger ensemble sizes to decode larger QLDPC codes.
This limitation is fundamental to the structures of the QLDPC codes.
An interesting new direction will be to study whether we can improve the noise-perturbed synthetic soft information to mitigate such cases.
Further, the NAED can be tested with frameworks such as the `decoder switching' \cite{toshio2025decoder}, to serve the higher degree, higher distance QLDPC codes.
Also, further studies can be done with different QLDPC codes to observe which QLDPC codes are more vulnerable to errors that cannot be spanned by the Tanner forests.
This can also be helpful in strategizing optimum post-processing solutions for the NAED framework.

\section{Code Availability}
The NAED decoder module developed in this work will be available publicly in the near future.
The numerical data generated from the time benchmarks and quantum memory experiments are available from the authors on reasonable requests.

\section{Acknowledgments}
M.B. acknowledges support
from the doctoral research fellowship of IISER Bhopal.
A.R. acknowledges funding support from the Department of Science and Technology (DST), Government of India, under the National Quantum Mission (NQM), DST/QTC/NQM/QComm/2024/2.
A.R. is thankful for the grant received from the U.S.—India Science and Technology Endowment Fund (USISTEF), USISTEF/QT/165/2023. 
\appendix

\section{CheckUpdate: cheapest column (data qubit) assignment using cheapest parity flip}
\label{ap:parity-dp}
The two-state dynamic parity program discussed in section \ref{sec:upward-pass} ensures an exact minimization of the cost for column (data qubit) assignments of the sub-tree extending from any detector (check) node.
We are interested in the minimum cost assignment.
This allows a more specialized implementation of the two-state dynamic programming we discussed.
We describe the algorithmic steps in Algorithm \ref{alg:parity-dp} and justify the implementation in this appendix.
The motivation for this more relaxed implementation is the fact that binary parity is the only constraint in the DP program.
Therefore, we can simply compute the optimal bit value $z_q$ for each child column (data qubit) $q \in \mathtt{ch}(c)$ of a detector (check) node $c$, and later, if the binary constraint is not satisfied through this optimum local assignment, we look for the cheapest assignment change.
Each child column (data qubit) $q \in \mathtt{ch}(c)$ can assume two possible assignments, i.e. $x_q \in \{0,1\}$.
The cost of any such local assignment is given by the information carried through the corresponding upward message $M_{q \rightarrow c}(x_q)$.
Therefore, the least costly assignment $z_q$ for any child column (data qubit) $q$ of the detector (check) node $c$ can be determined through:
\begin{align*}
    z_q = \arg \min_{x_q \in \{0,1\}}M_{q \rightarrow c}(x_q).
\end{align*}
The total cost of such locally individual minimum cost assignments of all the child columns (data qubits) is:
\begin{align*}
    C_{\mathrm{min}} = \sum_{q = 1}^{|\mathtt{ch}(c)|}M_{q \rightarrow c}(z_q),
\end{align*}
where $|\mathtt{ch}(c)|$ is the total number of the child column (data) nodes adjacent to detector (check) node $c$.
\begin{lemma}
    Any optimum child column (data qubit) assignment obtained by assigning individually minimum cost assignments of the columns (data qubits) can be adjusted to the required parity through the single column (data qubit) flip having the least cost.
\end{lemma}
\begin{proof}
    Consider the individual least cost assignments generate the child column (data qubit) parity $y_{\mathrm{child}} = z_0 \oplus z_1 \oplus \cdots \oplus z_{|\mathtt{ch}(c)|}$.
    The received syndrome bit of the detector (check) node $c$ is $s_c$, which suggests if the parent column (data qubit) $p$ of detector (check) node $c$ assumes value $x_p$, the required parity from the child column (data qubit) nodes is
    \begin{align*}
        y_{\mathrm{req}} = s_c \oplus x_p.
    \end{align*}
    Suppose a set of child columns (data qubits) $S = \{q : x_q \neq z_q\}$ are flipped from their individual least cost assignments (i.e. $z_q \to x_q$).
    These assignment flips change the total assignment costs of the child columns (data qubits), as follows:
    \begin{align*}
        \Delta = C_{\mathrm{min}} + \sum_{q \in S}\delta_q,
    \end{align*}
    where $\delta_q$ is the cost of flipping individual child column (data qubit)s from their local optimum value, and it is defined to be $\delta_q = M_{q \rightarrow c}(1 - z_q) - M_{q \rightarrow c}(z_q)$.
    Further, as $z_q$ denotes the minimum cost assignment of child column (data) node $q$, $\delta_q$ is always positive, i.e., $\delta_q \geq 0$.
    Now there are two cases to consider:
    \begin{itemize}
        \item The individually optimum child column (data) node assignments satisfy the required parity, i.e. $y_{\mathrm{child}} = y_{\mathrm{req}}$:
        In this case, the existing assignment is optimum as any set of even child node flips will increase the cost of the assignment, and any set of odd flips will result in the change of the parity.
        This implies that the locally optimal assignment is also the globally optimal parity-constrained assignment for the underlying detector (check) node.
        \item The individually optimum child assignments do not satisfy the required parity, i.e. $y_{\mathrm{child}} \neq y_{\mathrm{req}}$:
        In these circumstances, the local optimal assignments are infeasible due to the received syndrome parity, and a set of child columns (data qubits) needs to be flipped.
        This further suggests that the set $S$ must have an odd number of candidates.
        However, we have seen $\delta_i \geq 0$.
        Therefore, a single element in the set $S$ is sufficient and is optimally required to incur minimum penalty to the minimum assignment cost $C_{\mathrm{min}}$, i.e.
        \begin{align*}
            S^{*} = \{q^{*}\} : q^{*} = \arg \min_q \delta_q.
        \end{align*}
        Therefore, the required parity-constrained optimal child column (data qubit) assignment should incur an additional flip of
        \begin{align*}
            z_{q^{*}} \gets z_{q^{*}} \oplus 1.
        \end{align*}
        The constrained minimum cost, as well as the message sent from the detector (check) node $c$ to its parent column (data qubit) node $p$ is:
        \begin{align*}
            M_{c \rightarrow p} = C_{\mathrm{min}} + \delta_{q^{*}}.
        \end{align*}
        This completes our claim that a single flip from the unconstrained individual minimum cost assignments of the child columns (data qubits) is sufficient to impose the syndrome constraint, and the resultant child column (data qubit) assignment represents the least cost assignment satisfying the binary constraint of the syndrome bit.
    \end{itemize}
\end{proof}
This specialized implementation of $\mathtt{CheckUpdate}$ is exactly equivalent to the DP program discussed previously.
For any detector (check) node, the parity DP stage computes the child sub-tree column (data qubit) assignment with the least cost given by:
\begin{align*}
&M_{c \rightarrow p}(x_p)
=
\min_{x_q: q\in\mathtt{ch}(c)}
\sum_{q \in \mathtt{ch}(c)}
M_{q \rightarrow c}(x_q),\\
&\text{such that}\,\, \bigoplus\limits_{q\in\mathtt{ch}(c)}x_q
    = y_{\mathrm{req}}.
\notag
\end{align*}

In this section, we discussed a method to obtain the constrained individual minimum cost assignment, which exactly computes this sum without maintaining a two-state message and computing the two dynamic parities.
Instead, the specialized implementation simply chooses an unconstrained individual minimum cost assignment, implementing the unconstrained version of the above equation, and later flips a single column (data qubit) assignment.
This reduces the storage and compute overhead.
The $\mathtt{CheckUpdate}$ uniquely implements the above equation with the cheapest flip implementation by computing the same minimization in closed form:
\begin{align}
    M_{c\to p}(x_p)
    =
    \begin{cases}
    C_{\mathrm{min}},
    &
    y_{\mathrm{child}} = y_{\mathrm{req}},
    \\[0.5em]
    C_{\mathrm{min}}+\min\limits_q \delta_q,
    &
    y_{\mathrm{child}} \neq y_{\mathrm{req}}.
    \end{cases}
\end{align}

\section{Syndrome spanning capabilities of a Tanner Forest}
\label{ap:syndrome-forest}
We now discuss that the forests constructed using sparsified DEMs guarantee syndrome span for the surface codes, but it cannot be guaranteed for the BB codes.
The BB code parity check matrix has $d_v = 3$ regular data qubits.
We mentioned before that the sparse DEM construction from Ref. \cite{demarti2026almost} maps the full DEM into a sparser form where the maximum column weight of the sparse DEM equals the maximum column weight of the parity check matrix of the underlying QLDPC code.
For BB codes, the sparse DEM therefore has columns with maximum weight $3$.
This puts a fundamental restriction on the construction of the Tanner forests that can span the syndrome for BB codes.\\\\
Before arguing these facts, we show the condition that ensures the acyclic property of the forest $F$ during any new addition of a column (data qubit) from the sparse DEM (or, PCM).
Cyclomatic number of an undirected graph $G(V, E)$ counts the number of cycles, which is defined as
\begin{align}
    \beta(G) = |E| - |V| + C,
\end{align}
where $C$ is the number of connected components in the graph.
Suppose a new column of the sparse DEM matrix $D_{:q}$ is soon to be added into the forest $F$.
The weight of the column is the cardinality of its support, i.e. $\mathtt{wt}(D_{:q}) = |\mathtt{supp}(D_{:q})|$.
Before the addition of column $q$, the forest satisfies
\begin{align*}
    \beta(F) = |E_F| - |V_F| + C_F = 0.
\end{align*}
Adding one new column increases $|V_F| \gets |V_F| + 1$, $|E_F| \gets |E_F| + \mathtt{wt}(D_{:q})$.
The $\mathtt{wt}(D_{:q})$ adjacent detector nodes of column $D_{:q}$ belong to $t_q$ separate components before the addition.
Therefore, the addition of $D_{:q}$ into $F$ merges all those components into a single one, i.e. the change in the number of components of the forest F is $C_F \gets C_F - t_q + 1$.
This means the change in cyclomatic number is
\begin{align}
    \Delta \beta(F) = \mathtt{wt}(D_{:q}) - t_q.
    \label{eq: change-cyclomatic-num}
\end{align}
Eq. \eqref{eq: change-cyclomatic-num} implies that the acyclicity of the forest $F$ can only be maintained if 
\begin{align}
    \mathtt{wt}(D_{:q}) = t_q.
    \label{eq: forest-criterion}
\end{align}
This condition means that all the adjacent detectors of column $D_{:q}$ lie in distinct components before the addition.
\\\\
Assume that a forest $F$ is constructed from the sparse DEM of a QLDPC code.
The sparse DEM matrix $H_{\mathrm{sdem}} \in \mathbb{F}^{m \times n}_2$ of the QLDPC code is an $m \times n$ matrix, where $m$ is the number of detectors and $n$ is the total number of error mechanisms.
For a surface code, the columns of $H_{\mathrm{sdem}}$ have a maximum weight of $2$.
We argue that processing the columns of $H_\mathrm{{sdem}}$ in any arbitrary order to construct the forest always ensures $\mathrm{IM}(H_{\mathrm{F}}) = \mathrm{IM}(H_{\mathrm{sdem}})$. 
This property is independent of any methods used for the reliability ordering, i.e., BP-based or noise-perturbed synthetic soft information-based.
Consider any weight $1$ column $D_{:q}$ from $H_{sdem}$.
Such columns touch only a single detector, therefore $\mathtt{wt}(D_{:q}) = t_q = 1$ always satisfy  Eq. \eqref{eq: forest-criterion}.
On the otherhand consider a weight two column with support $\mathtt{supp}(D_{:q}) = \{d_1, d_2\}$.
A rejection of $D_{:q}$ might happen if $d_1$ and $d_2$ already belong to the same connected component of the forest.
However, this implies another aspect.
This component is connected and acyclic, so there must be a unique path from $d_1$ to $d_2$.
Any error mechanism supported on this path is equivalent to a detector violation of $\{d_1, d_2\}$.
Therefore, the syndrome of $D_{:q}$ remains protected.
The effect of the error mechanism $D_{:q}$ is replaced by a path connecting $d_1, d_2$; i.e., the syndrome of every rejected weight-two column lies in the span of the retained forest columns.
Therefore, it implies $\mathrm{IM}(H_{\mathrm{F}}) = \mathrm{IM}(H_{\mathrm{sdem}})$.
This argument does not assume any specific ordered construction of the forest.
Therefore, the graph-like error mechanisms of sparse surface code DEM ensure syndrome-compatible forest constructions in any general scenario.\\\\
We now consider a similar argument for the BB codes.
Consider a weight $3$ column $D_{:q}$ of the BB code sparse DEM that has been rejected during the forest construction.
Let the support of the rejected column be $\mathtt{supp}(h_q) = \{d_1, d_2, d_3\}$, and the rejection happened because at least two of the adjacent detectors $d_1$ and $d_2$ belong to the same connected component, thereby violating the condition of Eq. \eqref{eq: forest-criterion}.
Similar to the previous argument, there should be a unique path between $d_1$ and $d_2$.
However this path spans the syndrome $\{d_1, d_2\}$, not the $\mathtt{supp}(h_q)$.
Therefore, non-graph-like error mechanisms in general do not guarantee that every syndrome can be spanned by a specific ordered forest construction, i.e., in general $\mathrm{IM}(H_\mathrm{F}) \neq \mathrm{IM}(H_{\mathrm{sdem}})$.
In simpler terms, a rejected cycle formation during the Tanner forest construction might remove an independent error mechanism and its syndrome support.

\begin{figure*}[t]
    \centering

    \begin{minipage}[t]{0.48\textwidth}
        \centering
        \includegraphics[width=\linewidth]{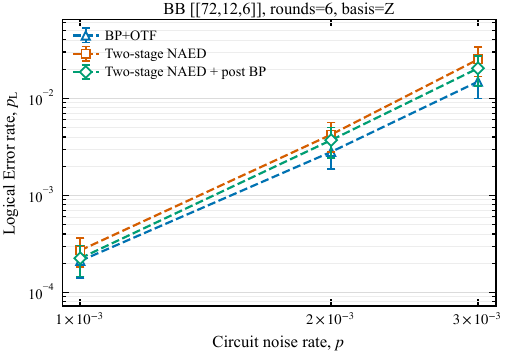}
        \par\smallskip
        \textbf{(a)}
    \end{minipage}
    \hfill
    \begin{minipage}[t]{0.48\textwidth}
        \centering
        \includegraphics[width=\linewidth]{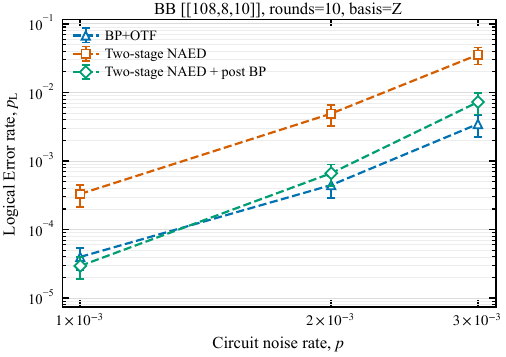}
        \par\smallskip
        \textbf{(b)}
    \end{minipage}

    \vspace{0.8em}

    \begin{minipage}[t]{0.50\textwidth}
        \centering
        \includegraphics[width=\linewidth]{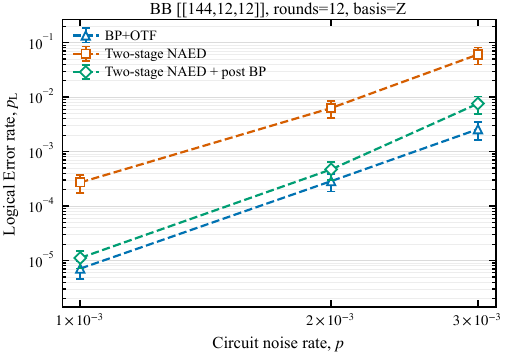}
        \par\smallskip
        \textbf{(c)}
    \end{minipage}

    \caption{
    A circuit-level LER comparison between the two-stage NAED + post-BP and BP+OTF decoder with decimation strategy for the BB codes.
    }
    \label{fig:bb-ler-comp-naed-otf}
\end{figure*}

\section{Circuit level simulation details}
\label{ap:circuit-simulation}
In this appendix, we test the performance of surface and bivariate bicycle codes using a quantum memory experiment.
The noise model we assume is a standard circuit-level Pauli noise model.
It is one of the most realistic noise models used for various benchmarks throughout the literature \cite{geher2024error}.
In general, it assumes that a fault can occur in any component of a stabilizer measurement circuit.
We use the following model specifications for the circuit-level noise:
\begin{itemize}
    \item A single qubit depolarizing error after each single qubit gate, with probability $p$, i.e. $p_X = p_Y = p_Z = \frac{p}{3}$.
    \item Each two-qubit entangling gate is followed by a two-qubit depolarizing error with probability $p$.
    Therefore, each error in $\{I, X, Y, Z\}^{\otimes 2}$ can happen with probability $\frac{p}{15}$.
    \item Each measurement outcome suffers a bit flip with probability $p$.
    \item In each layer of stabilizer circuit, the idling qubits suffer a depolarizing noise with probability $p$, i.e. $p_X = p_Y = p_Z = \frac{p}{3}$.
    \item The prepared initial qubit states are flipped to their orthogonal state with probability $p$.
\end{itemize}
Under this noise model, a single error-correction round is insufficient to achieve good error recovery.
A standard practice is to perform $d$ rounds of error correction for a distance $d$ code.
However, fault-tolerant quantum memory also requires a good measurement schedule for the stabilizers.
We use the rotated memory circuits from \textit{Stim} \cite{gidney2021stim} for the surface code.
For BB codes, we use the circuits from the GitHub repository of the \textit{SlidingWindowDecoding} \cite{gong2024toward}.
Both circuit constructions achieve good fault-tolerant performance under the existing decoders like BP+OSD$0$.\\\\
Circuit-level simulation captures the space-time picture of the errors.
It contains severe effects of error propagation through the components of stabilizer circuits that cannot be captured by the CSS parity-check matrix of a QLDPC code.
The decoding problem for circuit-level simulations is solved over the detector error model (DEM) matrix.
A plethora of research includes the impact of DEM on circuit-level simulations.
The reader can refer to the Refs. \cite{derks2025designing, gidney2021stim, higgott2025sparse}.
A more concise version of decoding circuit-level noise using DEM can be found in Section II. C. of Ref. \cite{demarti2026almost}.
In the benchmark plots comparing various decoder performances, we show the estimates of the logical error rates at different circuit noise error rates.
A logical error occurs when a logical observable obtained from the decoder prediction does not match the actual one.
We perform standard Monte-Carlo simulations to get a good estimate of these logical error rates.

\begin{figure}
    \centering
    \includegraphics[width=0.5\textwidth]{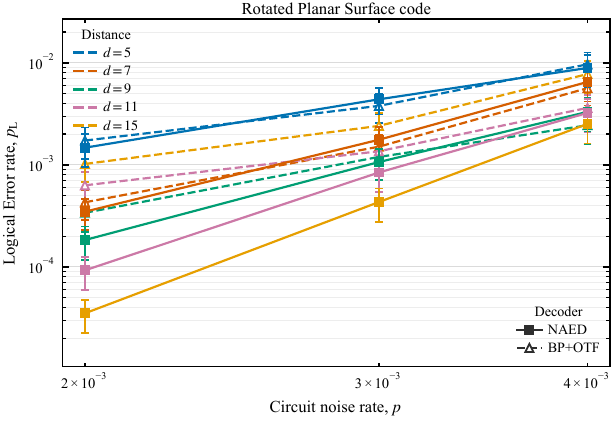}
    \caption{A circuit-level LER comparison between the BP+OTF decoder without decimation and NAED for Rotated planar surface codes.
    }
    \label{fig:surf-ler-comp-naed-bpotf}
\end{figure}

\section{Comparing NAED with BP+OTF Decoder}
We now compare the decoding performances of the NAED with the BP+OTF decoder.
We use BP+OTF to denote the BP+BP+OTF workflow of the decoder mentioned in Algorithm $2$ of Ref. \cite{demarti2026almost}.
For the BP+OTF decoder, we use the same configuration detailed in the reference for BB code simulations.
An ensemble of decoders is created by varying the maximum number of BP iteration caps from $17$ to $391$ in steps of $17$, i.e., a different instance from the ensemble is determined by the maximum BP iteration cap assigned to its initial BP round.
This creates a set of $23$
decoding instances in the ensemble.
Further, the sparse BP round and the BP round on OTF are capped at $113$ iterations (see section VI of Ref. \cite{demarti2026almost}).
The BP+OTF decoder with this set of configurations is compared with the two-stage NAED + post-BP ensemble having the same configuration described in Fig. \ref{fig:bb-ler-comparison}.
The LER comparisons for both decoders have been shown in Fig. \ref{fig:bb-ler-comp-naed-otf}.
Note that the BP+OTF decoder here uses a decimation strategy in the final OTF+BP round.
We assign a decimation parameter of $10^{-9}$ similar to the Ref. \cite{demarti2026almost}.
The decimation strategy assigns a very small probability to the columns that does not belong to the forest.
However, for surface code, if we do not use the decimation strategy and consider complete removal of the column nodes from the Tanner forest and employ the final BP round on a Tanner forest derived by removing the column nodes from the Tanner graph, the BP+OTF significantly degrades in performance.
For the surface code, we show a circuit-level LER plot comparison for the BP+OTF without decimation and NAED in Fig. \ref{fig:surf-ler-comp-naed-bpotf}.
The key reason behind the degraded performance of the BP+OTF decoder under no-decimation is the fact that a decimation parameter set to $0$ for the surface code essentially portrays the non-OTF columns as known nodes with $0$ value assigned to them for the neighboring detector nodes in the message passing rounds of BP.
However, removing them completely skips the information carried by these `zero' valued non-OTF nodes entirely.
However, under the exact inference algorithm of NAED, setting the non-OTF nodes to a $0$ value and removing them entirely is equivalent due to the nature of optimization carried out by the dynamic program of Algorithm \ref{alg:exact-map}.



\bibliography{references}

\end{document}